\documentclass[11pt]{article}
\usepackage{graphicx}
\usepackage{epstopdf}
\usepackage{amssymb,bm}
\usepackage{bbm}
\usepackage{amsmath}
\usepackage{amsfonts}
\usepackage[margin=1.0in]{geometry}
\usepackage{cite}
\usepackage[shortlabels]{enumitem}
\usepackage{comment}

\usepackage{color}

\DeclareMathOperator{\SNR}{\mathsf{SNR}}

\allowdisplaybreaks

\usepackage{amsthm}
\newtheorem{theorem}{Theorem}
\newtheorem*{theorem*}{Theorem}
\newtheorem{lemma}{Lemma}

\theoremstyle{definition}  
\newtheorem{definition}{Definition}
\newtheorem{remark}{Remark}
\newtheorem{corollary}{Corollary}

\usepackage{chngcntr}
\counterwithout{theorem}{section}
\counterwithout{definition}{section}
\counterwithout{lemma}{section}
\counterwithout{remark}{section}
\counterwithout{assumption}{section}
\counterwithout{proposition}{section}
\counterwithout{corollary}{section}
\counterwithout{claim}{section}

\begin{document}
\title{Fundamental Limits of Wireless Caching Under Mixed Cacheable and Uncacheable Traffic
\footnotetext{This work was supported in part by the European Research Council (ERC) under the ERC grant agreement N. 789190 (project CARENET),
and the ERC grant agreement N. 725929 (project DUALITY).
This paper was presented in part at the 2020 IEEE International Symposium on Information Theory.
\\
H. Joudeh was with the Faculty of Electrical Engineering and Computer Science, Technische Universit\"{a}t Berlin, 10587 Berlin, Germany.
He is now with the  Department of Electrical Engineering, Eindhoven University of Technology, 
5600 MB Eindhoven, The Netherlands (e-mail: h.joudeh@tue.nl).
E. Lampiris  was with the Faculty of Electrical Engineering and Computer Science, Technische Universit\"{a}t Berlin, 10587 Berlin, Germany.
He is now with the Communication Systems Department, EURECOM, 06410 Sophia Antipolis, France (e-mail: lampiris@eurecom.fr).
P. Elia is with the Communication Systems Department, EURECOM, 06410 Sophia Antipolis, France (e-mail: elia@eurecom.fr).
G. Caire is with the Faculty of Electrical Engineering and Computer Science, Technische Universit\"{a}t Berlin, 10587 Berlin, Germany (e-mail: caire@tu-berlin.de). 
}}
\author{Hamdi~Joudeh, Eleftherios~Lampiris, Petros~Elia and Giuseppe~Caire}
\date{}
\maketitle
\begin{abstract}
We consider cache-aided wireless communication scenarios where each user
requests both a file from an a-priori generated cacheable library (referred to as `content'), 
and an uncacheable `non-content' message generated at 
the start of the wireless transmission session.
This scenario is easily found in real-world wireless networks, where the two 
types of traffic coexist and share limited radio resources.
We focus on single-transmitter, single-antenna wireless networks with cache-aided receivers, where the wireless channel is modelled by a degraded Gaussian broadcast channel (GBC).
For this setting, we study the delay-rate trade-off, which characterizes the content delivery time and non-content communication rates that can be achieved simultaneously. 
We propose a scheme based on the separation principle, which isolates the coded caching and multicasting problem from the physical layer transmission problem.
We show that this separation-based scheme is sufficient for achieving an information-theoretically order-optimal performance, 
up to a multiplicative factor of $2.01$ for the content delivery time, when working in the generalized degrees of freedom (GDoF) limit.
We further show that the achievable performance is near-optimal after relaxing the GDoF limit, up to an additional additive factor of $2$ bits per dimension for 
the non-content rates.  
A key insight emerging from our scheme is that in some scenarios considerable amounts of non-content traffic can be communicated while maintaining the minimum 
content delivery time, achieved in the absence of non-content messages; compliments of `topological holes' arising from 
asymmetries in wireless channel gains.
\end{abstract}
\newpage
\section{Introduction}
\label{sec:introduction}
Cache-aided architectures have emerged as an essential next step in the evolution of communication networks \cite{Paschos2018}. 
This is backed by two key factors: the explosion in \emph{cacheable} data traffic due to on-demand access to internet content (e.g. video-streaming); and the low cost and ubiquity of large on-board storage memory. 
In the caching paradigm, popular content is pro-actively stored across network nodes during off-peak times, when network resources are underutilized, and then the pre-stored content is leveraged to alleviate the traffic load during congested peak times \cite{Shanmugam2013}. 

The recent few years have seen the emergence of information-theoretic  studies that aim at establishing the fundamental limits of communication over cache-aided networks. These studies have been initiated by the seminal work of Maddah-Ali and Niesen in \cite{Maddah-Ali2014}.
For an idealized symmetric broadcast channel (BC), in which cache-equipped users (receivers) are connected to a server (transmitter) through a noiseless shared link, Maddah-Ali and Niesen showed that a novel  cache-aided 
\emph{coded-multicasting} scheme 
can serve an arbitrarily large number of users with finite resources (e.g. time and bandwidth).
The achievable performance in \cite{Maddah-Ali2014}, characterized in terms of the shared link \emph{normalized load},\footnote{This can also be seen as a \emph{normalized delivery time} (NDT) measure, where one unit of time (i.e. time slot) is equivalent to the time required to deliver a single file in the absence of caches.} was shown to be order-optimal in the information-theoretic sense,  maintaining a constant multiplicative factor from the optimal performance.  
The information-theoretic optimality result in \cite{Maddah-Ali2014} was tightened later on in \cite{Wan2016,Yu2018} under the restriction of uncoded cache placement, and in \cite{Yu2019} for the general unrestricted case. 
\subsection{Wireless caching}
The bulk of data traffic nowadays is generated for wireless and mobile devices, a trend foreseen to continue and grow in the forthcoming years.
This has driven a surge of interest in extending the information-theoretic coded caching paradigm to wireless network settings.
Such settings differ from their idealized counterparts (e.g. \cite{Maddah-Ali2014,Shariatpanahi2016}) in several important aspects, which most notably include: the noisiness of wireless channels; the uneven and asymmetric nature of wireless network topologies; 
and the crucial impact of fading and channel state information (CSI) feedback. 

In the context of single-transmitter networks, the coded caching paradigm has been extended to noisy settings, 
including erasure and degraded BCs \cite{Ghorbel2016,Amiri2018a,Amiri2018b,Amiri2018,Salman2019,Bidokhti2017a,Bidokhti2018,Zhang2017a,Lampiris2019}, and multi-antenna 
BCs \cite{Zhang2015,Zhang2017,Piovano2017,Ngo2018,Lampiris2018a,Piovano2019,Bergel2018,Shariatpanahi2019,Cao2019}.
For multi-transmitter settings, coded caching has been studied in device-to-device (D2D) networks \cite{Ji2016}, interference networks with caches at the transmitters only or at both  transmitters and receivers 
\cite{Maddah-Ali2015,Naderializadeh2017,Xu2017,Cao2017,Hachem2018,Lampiris2017,Piovano2020}, 
and fog radio access networks (F-RANs)  \cite{Sengupta2017,Zhang2019}, among other settings.
The interplay between CSI feedback and coded caching in multi-antenna and multi-transmitter networks has been investigated in \cite{Zhang2015,Zhang2017,Piovano2017,Ngo2018,Lampiris2018a,Piovano2019,Lampiris2017,Piovano2020}.
Moreover, some recent works explore the role of multi-antenna transmitters and shared receiver caches in alleviating the subpacketization complexity bottleneck of coded caching and multicasting schemes \cite{Lampiris2018,Parrinello2020}. 
While many of the above works focus on simpler symmetric topologies, aspects specific to uneven topologies are of paramount importance, as various studies show  \cite{Amiri2018b,Amiri2018,Salman2019,Bidokhti2017a,Bidokhti2018,Zhang2017a,Lampiris2019,Ngo2018,Bergel2018}.
\subsection{Mixed cacheable and uncacheable traffic}
All the above-mentioned works consider scenarios in which the network carries a single
type of traffic that takes the form of (popular) content drawn from an a-priori generated library (or database).
This approach has been very useful and successful in gaining insights into the 
fundamental limits of cache-aided wireless networks, and the design of optimal and near optimal caching and coded multicasting schemes.
Nevertheless, wireless data traffic does not comprise of only \emph{cacheable} content.
\emph{Uncacheable} (non-content) traffic, generated from a plethora of interactive or real-time applications, as well as voice and video calls, to name a few examples, also constitutes  
a significant portion of overall wireless data traffic (estimated as $40$ percent \cite{Paschos2016}). 
Moreover,  content popularity profiles in reality are far from static and may
change on a daily or even hourly basis \cite{Paschos2018,Paschos2016}. 
Therefore, newly generated content can be both in high demand as well as not yet registered or available in caches. 

Motivated by this mixed nature of data traffic, we initiate the study of cache-aided wireless networks with both \emph{cacheable}  and \emph{uncacheable}
types of traffic.
Throughout this work, we use \emph{content} traffic to describe pre-generated \emph{cacheable} files; and \emph{non-content} traffic to describe \emph{uncacheable} messages, instantaneously generated at the start of wireless transmission sessions. 
These definitions are further clarified below and in Section \ref{sec:system model}, where the problem setting is formally described. 
\subsection{Considered setting and adopted approach}
\label{subsec:intro_setting}
The setting of focus is a cache aided degraded Gaussian BC (GBC), comprising a single transmitter and $K$ receivers.
The transmitter has access to a pre-generated content library of $N$ equal-size files, while each user has a cache memory that can store 
content of size equal to the size of $M$ files during a \emph{placement phase}, which takes place well before communication sessions commence. 
The normalized cache size is defined as $\mu \triangleq \frac{M}{N}$, where 
$\mu \in [0,1]$. 
At the beginning of a communication session, known as a \emph{delivery phase}, 
each user requests a content file, as well as an instantaneously generated non-content message, not known a-priori  to the transmitter.
The setup is illustrated in Fig. \ref{fig:system_model}.
\begin{figure}[h]
\centering
\includegraphics[width = 0.55\textwidth]{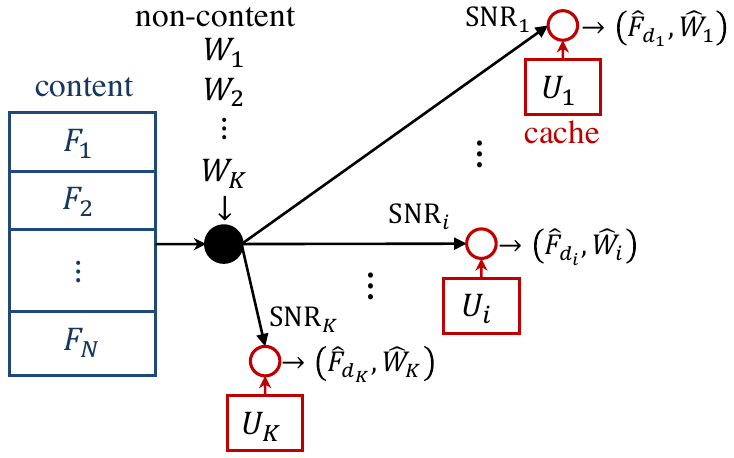}
\caption{\small $K$-user cache-aided degraded GBC with a content library and non-content messages.}
\label{fig:system_model}
\end{figure}

To gain initial insight, let us consider the special case of a symmetric physical channel with equal signal-to-noise ratio (SNR) across users.
We start by taking non-content messages out of the picture, and focus on the delivery of $K$ distinct content files.\footnote{Here we assume that $N \geq K$ for ease of exposition. This assumption is relaxed further on.}
In this case, the original Maddah-Ali and Niesen coded caching and multicasting scheme  (referred to as the MN scheme henceforth) \cite{Maddah-Ali2014}, coupled with a standard channel coding argument, achieves a (per-bit) communication 
delay of
\begin{equation}
\label{eq:delay_intro_only_content}
\mathcal{T} = \frac{1}{\log(1 + \SNR)} \cdot \frac{K (1 - \mu) } {1 + K \mu}.
\end{equation}
The above measure of delay corresponds to the number of (physical) channel uses required to deliver one bit of content for each user,
in the Shannon limit (i.e. as the file size approaches infinity).

Now consider the additional transmission of $K$ non-content messages, each intended to a unique user, at possibly distinct communication rates of $R_{1}, \ldots, R_{K}$.
The physical channel is now shared between multicast messages (coded content) and unicast messages (non-content), such that 
\begin{equation}
\label{eq:capacit_region_phy_K_user_sym}
R_{[K]} + \sum_{i = 1}^{K} {R}_{i} \leq \log(1 + \SNR)
\end{equation}
where $R_{[K]}$ denotes the rate of communicating multicast messages over the wireless channel.
For any feasible tuple of non-content rates $(R_{1}, \ldots, R_{K})$, one can readily achieve a communication delay of\footnote{We are interested in the total delay of the entire communication session, i.e. the time required to deliver both content files and non-content messages, in channel uses normalized by the number of file bits---see Section \ref{subsec:codes_rate_delay}.} 
\begin{equation}
\label{eq:delay_intro_non_content}
\mathcal{T} (R_{1},\ldots , R_{K}) = \frac{1}{\big( \log(1 + \SNR) - \sum_{i = 1}^{K} R_{i} \big)} \cdot \frac{K (1 - \mu) } {1 + K \mu}.
\end{equation}
The delay-rate trade-off in \eqref{eq:delay_intro_non_content} can be achieved  by employing 
a \emph{separation} approach, which separates the 
coded caching and multicasting problem from the physical channel transmission problem, see, e.g., \cite{Hachem2018}.
In particular, caching, generating coded multicast messages and recovering requested files from received coded messages and local cache contents are carried out at the bit level in the standard shared link fashion, using the MN scheme \cite{Maddah-Ali2014}.
On the other hand, in this same separation approach, the physical-layer sees a collection of multicast messages and unicast messages, and communicates them using standard channel coding.
As it turns out, for this particular \emph{symmetric} case, the trade-off in \eqref{eq:delay_intro_non_content} is order-optimal, i.e. within a constant multiplicative factor from the information-theoretic trade-off (see Section \ref{sec:converse}).

The reader may have noticed that as far as the physical channel is concerned, the rates in \eqref{eq:capacit_region_phy_K_user_sym} can be achieved using time-sharing, i.e. multicast and unicast messages are mapped into independent signals, communicated sequentially over distinct time slots. 
As one would imagine, the sufficiency of time-sharing in this setting is by virtue of symmetry.
In general, time-sharing  is rendered suboptimal by the
superposition and asymmetric nature of wireless channels, epitomized through the degraded GBC---one may envisage the superiority of superposition coding 
in asymmetric scenarios.
Nevertheless,  the current treatment in the literature of delivering content and non-content traffic
as two independent problems, necessarily 
leads to time-sharing-like schemes, where the two types of traffic are scheduled  on orthogonal physical-layer resource blocks. 
As we will see in this paper, this orthogonalization is suboptimal in general, specifically in asymmetric settings.

We propose to treat the two problems jointly.
In particular, while we maintain a separation architecture that isolates coded caching and multicasting from channel coding and physical-layer transmission, 
the transmission of messages corresponding to content and non-content traffics over the physical channel is carried out in a joint (non-orthogonal) manner, by leveraging power control with superposition coding and successive decoding.
This leads to an order-optimal performance in the information-theoretic sense, as we later show in this paper. 
\subsection{Generalized degrees of freedom regime}
Attempting to settle the above question by pursuing \emph{exact} delay-rate trade-off characterizations is bound to yield intricate solutions, which are not necessarily malleable for analysis or useful for gaining practical insights. 
To see this, let us consider a simple setting with $K = N = 2$ and $M = 1$.
In this case, we know from the MN scheme that the delivery of one 
coded multicast message of normalized size given by $\frac{K (1 - \mu) } {1 + K \mu} = 1/2$ is sufficient to satisfy distinct user demands. 

Using the above-described separation approach,  a delay-rate trade-off of
\begin{equation}
\label{eq:intro_delay_2_user}
\mathcal{T}(R_{1},R_{2}) = \frac{1}{ R_{12} } \cdot \frac{1} {2}
\end{equation}
can be achieved for any non-negative rate tuple $(R_{1},R_{2},R_{12})$ that satisfies 
\begin{equation}
\label{eq:capacit_region_phy_2_user}
\begin{aligned}
R_{12} + R_{1} & \leq  \log \left(1 + \frac{q\SNR_{1}} {1 + \bar{q} \SNR_{1}} \right) \\
R_{2} & \leq  \log \big(1 +  \bar{q} \SNR_{2} \big)
\end{aligned}
\end{equation}
for some power control variables $q \in [0,1]$ and $\bar{q} \triangleq 1 - q$, under a unit average power constraint.
Note that in the above, we assume without loss of generality that $\SNR_{1} \leq \SNR_{2}$.
The inequalities in \eqref{eq:capacit_region_phy_2_user} characterize the capacity region of the $2$-user degraded GBC with unicast and multicast (i.e. common) messages. 
This region, and hence the delay achieved by separation in \eqref{eq:intro_delay_2_user}, 
crucially depend on the auxiliary power control variable $q$, and in general cannot be expressed explicitly in terms of fixed channel parameters only  (i.e. $\SNR_{1}$ and $\SNR_{2}$).
This dependency on auxiliary power control variable(s) becomes more problematic in larger networks with arbitrary $K$, 
where the physical channel communicates multiple nested sets of multicast messages, 
giving rise to delay-rate characterizations which are difficult to analyse.
Effects of this complexity are seen through previous results on coded caching in the degraded GBC, 
see, e.g., \cite{Amiri2018b,Amiri2018,Salman2019}.

In this work, we circumvent the above-described complexity issue by taking a step back from the \emph{exact} delay-rate trade-off, and instead pursuing an \emph{approximate} characterization based on the \emph{generalized degrees of freedom} (GDoF)  measure \cite{Etkin2008}.
In the GDoF sense, the capacity region of the physical channel in \eqref{eq:capacit_region_phy_2_user} reduces to all non-negative GDoF tuples $(r_{1},r_{2},r_{12})$ that satisfy
\begin{equation}
\label{eq:GDoF_region_phy_2_user}
\begin{aligned}
r_{12} + r_{1} & \leq  \alpha_{1} \\
r_{12} + r_{1} + r_{2} & \leq  \alpha_{2}
\end{aligned}
\end{equation}
where $\alpha_{1}$ and $\alpha_{2}$ are GDoF-type channel strength parameters  that correspond to 
$\SNR_{1}$ and $\SNR_{2}$, respectively (see Section \ref{subsec:phy_channel}).
The GDoF region in \eqref{eq:GDoF_region_phy_2_user} is a polyhedron, 
and has the desirable
property of admitting a reduced explicit description in terms of fixed channel parameters only (i.e. $\alpha_{1}$ and $\alpha_{2}$), 
without the need for auxiliary power control variables. 
From \eqref{eq:GDoF_region_phy_2_user}, the \emph{generalized normalized delivery time} (GNDT), i.e. the GDoF counterpart of the delay in \eqref{eq:intro_delay_2_user}, is  given by 
\begin{equation}
\label{eq:intro_GNDT_2_user}
\tau =  \max \left\{ \frac{1}{  ( \alpha_{1} - r_{1} ) } , \frac{1}{ ( \alpha_{2} - (r_{1} + r_{2})  )} \right\} \cdot \frac{1} {2}
\end{equation}
obtained from the MN scheme and \eqref{eq:GDoF_region_phy_2_user}  by observing that for any feasible non-content GDoF tuple $(r_{1},r_{2})$, a multicast GDoF of $r_{12} = \max \left\{ \frac{1}{  ( \alpha_{1} - r_{1} ) } , \frac{1}{ ( \alpha_{2} - (r_{1} + r_{2})  )} \right\}$ is achievable.
As we will see in Section \ref{sec:converse}, the simplicity of the linear inequalities in \eqref{eq:GDoF_region_phy_2_user} allows for a direct comparison with counterpart information-theoretic outer bounds, from which we prove order-optimality.

The explicit nature of the above GNDT-GDoF trade-off allows for deriving 
useful operational insights.
For instance, \eqref{eq:intro_GNDT_2_user} suggests that we can communicate a non-content message at a GDoF of $r_{2} \leq \alpha_{2} - \alpha_{1}$ to user $2$, while simultaneously maintaining the GNDT  achieved in the absence of non-content messages.
As we will see in Section \ref{sec:main_results}, the GNDT-GDoF trade-off allows us to precisely quantify the gains due to the asymmetry in channel strengths for an arbitrary number of users. 
Through these \emph{topological holes} that appear as a result of asymmetry, we can communicate non-content messages at no cost in content delivery time.
Moreover, we will also see that the GNDT-GDoF characterization leads to an approximate delay-rate characterization, up to a small  gap.
A detailed exposition of the main results and insights is given in Section \ref{sec:main_results}.
\subsection{Overview of contributions and related works}
We conclude this section by highlighting the contributions of this work and relationship to prior art.
In the main result of this paper (Theorem \ref{theorem:optimal_trade_off}, Section \ref{sec:main_results}), 
we obtain an achievable GNDT-GDoF trade-off for the cache-aided degraded GBC with mixed content and non-content traffic; 
and we prove that this trade-off is order-optimal in the information-theoretic sense, up to a multiplicative factor of $2.01$.
Furthermore, we show that the GNDT-GDoF characterization leads to a counterpart delay-rate trade-off, which is information-theoretically optimal up to a multiplicative factor of $2.01$ for the content delay, and an additive factor of $2$ bits (per dimension) for 
the non-content rates, at all finite SNR values (i.e. after relaxing the GDoF limit).

The achievability of our result is based on the separation principle, where the coded caching side of the problem is separated
from the physical-layer communication side.
This separation approach gives rise to a new physical-layer problem concerning the characterization of the GDoF and capacity regions of 
the $K$-user degraded GBC with unicast and multiple multicast message sets.
We give a complete characterization of the GDoF region of this channel, and its capacity region up to a constant additive gap (see Section \ref{sec:degraded_GBC_unicast_multicast}).
This result may be of interest in its own right.

The converse proof of our main result is based on a non-trivial augmentation of the argument by Yu et al. \cite{Yu2019}, proposed for the idealized shared link setting.
Guided by separation in the achievability scheme, we devise a sequence of steps that separate the information-theoretic bounds 
into a set of terms that capture the physical channel capacity, and a second set of terms that capture the load due to the content caching 
and delivery.
The former are bounded by extending classical properties of the degraded BC, while the latter are bounded by invoking techniques from 
\cite{Yu2019}.

\textbf{Related works:} As a special case of our result, we recover the previous result in \cite{Lampiris2019}, where a similar setting was 
considered in the absence of non-content messages. 
In addition to generalizing  \cite{Lampiris2019} to scenarios with both content and non-content traffic, our work improves upon this previous result in several ways.
First, our new converse leads to a tighter order-optimality result, reducing the multiplicative factor obtained in \cite{Lampiris2019} 
from $4.02$ to $2.01$.
Second, our result extends the achievability argument in \cite{Lampiris2019} to the case with non-integer normalized aggregate cache size $K \mu$, and the case with more users than files (i.e. $N < K$).
We show that for non integer $K \mu$, a direct application of the memory-sharing principle yields a strictly suboptimal GNDT, and a superior performance
is achieved by treating the physical-layer transmission problem as one with two nested sets of multicast messages.
To address the case of $N < K$, we base our achievability on the Yu, Maddah-Ali and Avestimehr  (YMA) scheme \cite{Yu2018,Yu2019}, in contrast to the MN scheme adopted in \cite{Lampiris2019}.
Third,  we take a few steps beyond \cite{Lampiris2019}, and refine and relax the GNDT-GDoF results to obtain 
approximate delay-rate characterizations, which hold at all finite SNR values.

Another set of closely related results for the cache-aided degraded GBC, in the absence of non-content messages, are found in 
\cite{Amiri2018b,Amiri2018,Salman2019}, where the exact delay measure (or its rate reciprocal) is considered instead of the GNDT approximation.
In \cite{Amiri2018b}, the authors  focus on minimizing the transmit power subject to a delay constraint---a dual to the more common problem of delay minimization subject to a transmit power constraint.
The scheme in \cite{Amiri2018b} can be seen as a special case of the scheme we propose here, after eliminating uncacheable non-content messages, and the derived achievable performance has the merit of exactness.
Nevertheless, the achievable delay characterization in \cite{Amiri2018b}  highly depends on auxiliary power allocation variables that require further optimization, rendering it less flexible for direct analysis compared to the GNDT characterization we obtain here---see \eqref{eq:intro_delay_2_user} and \eqref{eq:intro_GNDT_2_user}.  
Moreover, the outer bound in \cite{Amiri2018b} is restricted to uncoded placement schemes and there are no guarantees of order-optimality 
(examined numerically in \cite{Amiri2018b}).

In \cite{Amiri2018}, a setting with multi-layered content is studied, 
where each file is described by several independent layers  representing refinements of the same content (e.g. higher quality), and which are  communicated opportunistically depending on users' SNRs.
This multi-layered setting shares an important aspect with the mixed traffic setting we study here, i.e. the opportunity to exploit 
\emph{topological holes} to communicate additional (asymmetric) messages beyond (symmetric) content files. 
On the other hand, there are also key discrepancies, including the assumption that all file layers are cacheable, and the dependency of caching schemes on the wireless network topology in \cite{Amiri2018}. 
Moreover, the results in \cite{Amiri2018} inherit some of the limitations in \cite{Amiri2018b}, e.g. the inexplicit characterizations which are strongly coupled with auxiliary optimization variables, as well as the lack of  information-theoretic optimality guarantees. 
Finally, in \cite{Salman2019} the authors obtain a complete characterization of the optimal delay in the $2$-user cache-aided GBC under the restriction 
of uncoded caching schemes.
Nevertheless, it is not yet clear whether the techniques can be extended to more general setting, with an arbitrary number of users and 
(possibly) coded cache placement.
\subsection{Notation}
For positive integers $z_{1}$ and $z_{2}$, with $z_{1} \leq z_{2}$, the sets $\{1,2,\ldots,z_{1}\}$ and $\{z_{1},z_{1}+1,\ldots,z_{2}\}$ are  denoted by $[ z_{1} ]$ and $[ z_{1}:z_{2} ]$, respectively. $\binom{z_{2}}{z_{1}}$ denotes the binomial coefficient.
For a real number $a$, we denote $\max\{0,a\}$ by $(a)^{+}$.
The tuple 
$(a_{1}, \ldots, a_{y}, b_{1}, \ldots, b_{z})$ is denoted by $(a_{i} : i  \! \in \!  [y], b_{i} : i \!  \in \!  [z] )$.
The cardinality of set
$\mathcal{A}$ is denoted by $|\mathcal{A}|$.
For sets $\mathcal{A}$ and $\mathcal{B}$, the set of elements in $\mathcal{A}$ and not in $\mathcal{B}$ is denoted by $\mathcal{A} \setminus \mathcal{B}$.
For any $\mathcal{A} \subseteq \mathbb{R}^{K}$,  the closure of set $\mathcal{A}$ is denoted by  $\mathrm{cl}\{ \mathcal{A} \}$.
\section{Problem Setting}
\label{sec:system model}
In this section, we formally describe the system model introduced in Section \ref{subsec:intro_setting},
and then proceed to define the performance measures and formulate the problem.
As mentioned earlier, we consider a wireless network consisting of a single transmitter
(server) and $K$ receivers (users).
The transmitter has access to a content library of $N$ files, 
denoted by $F_{1},\ldots,F_{N}$, each of size $B$ bits.
Each user $k$ is equipped with an isolated cache memory of size $MB$ bits, where $M \in [0,N]$.
The network operates in two phases: a \emph{placement phase} and a \emph{delivery phase}.
\begin{enumerate}
\item \emph{Placement phase}: During this phase, users have access to the entire library of files to fill the content of their caches. This occurs without knowledge of future file requests.
\item \emph{Delivery phase}:  
Each user $k$ requests a \emph{content} file $F_{d_{k}}$,
where $d_{k} \in [N]$ is the corresponding demand index.
Moreover, the transmitter generates $K$ \emph{non-content} messages
$W_{1}, \ldots, W_{K}$, intended to users $1,\ldots,K$, respectively.
These messages are mutually independent, independent of the content library, and may vary in size.
During the delivery phase, the transmitter sends a codeword over the physical channel; while each user $k$ receives a corresponding noisy signal and tries to recover $\big(F_{d_{k}},W_{k}\big)$ from this signal and the local cache content. 
\end{enumerate}
\begin{remark} \textbf{(Files and Messages).}
\label{remark:files_messages}
The word ``\emph{files}" is used to describe $F_{1},\ldots,F_{N}$, which are pre-generated \emph{content} messages, known beforehand to the server 
and revealed to users during the placement phase.
Files represent predictable types of traffic, e.g. popular internet content.
On the other hand, the word ``\emph{messages}"  describes 
$W_{1},\ldots,W_{K}$, which are classical \emph{non-content} messages generated in real time just ahead of transmission during the delivery phase.
Messages represent unpredictable types of traffic, e.g. voice and video calls, or recent internet content.
\end{remark}
\subsection{Physical channel}
\label{subsec:phy_channel}
The physical channel is a $K$-user degraded GBC.
In the $t$-th use of the physical channel, where $t \in \mathbb{N}$, the  input-output relationship is described as:
\begin{equation}
Y_{k}(t) = h_{k} X(t) + Z_{k}(t).
\end{equation}
In the above, $X(t) \in \mathbb{C}$ is the input signal; while 
$Y_{k}(t),Z_{k}(t), h_{k} \in \mathbb{C}$ are the output signal, 
zero-mean, unit-variance additive white Gaussian noise (AWGN) signal, and the (fixed) 
channel coefficient of user $k$, respectively.
Communication  occurs over $T$ channel uses,
in which the transmitter is subject to a unit average power constraint given by
\begin{equation}
\label{eq:power_const}
\frac{1}{T} \sum_{t = 1}^{T} | X(t) |^{2}  \leq 1.
\end{equation}

For each user $k$, the SNR is determined by the corresponding channel coefficient, and is  given by $ \SNR_{k} = |h_{k}|^{2}$.
We assume, without loss of generality, that the following order holds:
\begin{equation}
\label{eq:order_SNR}
\SNR_{1} \leq \SNR_{2} \leq \cdots \leq \SNR_{K}.
\end{equation}
For GDoF (and GNDT) purposes, we express the SNR of each user $k$ as 
\begin{equation}
\SNR_{k} = P^{\alpha_{k}}
\end{equation}
where the exponent $\alpha_{k}$ is known as the \emph{channel strength level},
while $P > 1$ is a \emph{nominal power parameter} which approaches infinity to define
 the GDoF limit---see \cite{Piovano2019,Etkin2008,Bresler2010,Jafar2010}.
We assume, without loss of generality, that $\alpha_{k} \in (0,1]$ and $\alpha_{K} = 1$, which alongside the order in \eqref{eq:order_SNR} translate to
\begin{equation}
\label{eq:order_alpha}
0 < \alpha_{1} \leq \alpha_{2} \leq \cdots \leq \alpha_{K} = 1.
\end{equation}
The channel strength tuple is given by $\bm{\alpha} \triangleq (\alpha_{1} , \ldots , \alpha_{K})$.
\begin{remark}
\label{remark_GDoF_model}
Truncating channel strength levels such that $\alpha_{k} > 0$  translates to  $\SNR_{k}  > 1$ for all  users $k \in [K]$ (recall that $P > 1$), which is common practice in GDoF studies. 
This excludes scenarios where $\SNR_{i}  \leq 1$ for some users $i \in [K]$,
as such users receive their desired signals at the same level of noise (at best), and hence achieve zero GDoF.
In the constant-gap capacity sense,  $\SNR_{i}  \leq 1$ leads to an achievable rate which is bounded above by $1$ bit per channel use.
Therefore, an achievable rate region which excludes user $i$, e.g. by setting the corresponding rate to zero,
may still be within $1$ bit per channel use from the capacity region. 
\end{remark}
\subsection{Codes, Rates and Delay}
\label{subsec:codes_rate_delay}
Files $F_{1},\ldots,F_{N}$ are independent random variables, each uniformly distributed over the set $[2^{\lfloor B \rfloor}]$.
To define asymptotic limits, we scale the file size $B$ with the number of physical channel uses $T$ as $B = T R_{F}$, where $R_{F}$ is the content rate in bits per channel use.
Messages $W_{1},\ldots,W_{K}$ are also  independent random variables, yet not necessarily identical. Each $W_{k}$ is uniformly distributed over the set $[2^{ \lfloor TR_{k} \rfloor }]$, where $R_{k}$ is the corresponding  message rate and  $\mathbf{R} \triangleq (R_{1},\ldots,R_{K})$ denotes a message rate tuple.
A demand tuple is defined as $\mathbf{d} \triangleq  (d_{1},\ldots,d_{K}) \in [N]^{K}$.

A code $(T,R_{F},\mathbf{R},M)$ consists of the above file and message sets in addition to the following:
\begin{itemize}
\item  A caching strategy  $\bm{\phi} \triangleq (\phi_{1}, \ldots, \phi_{K})$,
comprising $K$ caching functions. Each caching function $\phi_{k}: [2^{\lfloor TR_{F} \rfloor}]^{N} \rightarrow [2^{\lfloor TR_{F} M \rfloor}]  $ is a map between the $N$ library files and 
the  cache content of the corresponding user $k$, denoted by $U_{k}$. That is
\begin{equation}
U_{k} = \phi_{k}(F_{1},\ldots,F_{N}).
\end{equation}
\item An encoding function $\psi: [N]^{K}  \times [2^{ \lfloor TR_{F} \rfloor }]^{N} \times [2^{\lfloor TR_{1} \rfloor}] \times \cdots \times [2^{\lfloor TR_{K} \rfloor}] \rightarrow \mathbb{C}^{T}$
which maps the demand tuple, $N$ files and $K$ messages to a codeword 
$X^{T} \triangleq \big(X(1), \ldots, X(T) \big)$, which satisfies the power constraint in \eqref{eq:power_const}.
In particular, we have 
\begin{equation}
X^{T}  = \psi (\mathbf{d},F_{1},\ldots,F_{N}, W_{1}, \ldots, W_{K}).
\end{equation}
\item A decoding strategy $\bm{\eta} \triangleq (\eta_{1}, \ldots, \eta_{K})$, comprising $K$ decoding functions.
Each decoding function $\eta_{k}: [N]^{K}  \times \mathbb{C}^{T}  \times [2^{\lfloor TR_{F} M \rfloor}]  \rightarrow [2^{\lfloor TR_{F} \rfloor}] \times [2^{\lfloor TR_{k} \rfloor}]$ maps the demand tuple, received signal $Y_{k}^{T} \triangleq \big(Y_{k}(1), \ldots, Y_{k}(T) \big) $ 
and local cache content to an estimate of $(F_{d_{k}} , W_{k}) $, i.e.
\begin{equation}
(\hat{F}_{d_{k}} , \hat{W}_{k})  = \eta_{k} \big(\mathbf{d}, Y_{k}^{T} , U_{k} \big).
\end{equation}
\end{itemize}
For any code $(T,R_{F},\mathbf{R},M)$,  the probability of decoding error is defined as
\begin{equation}
\label{eq:error_prob}
P_{e,T} \triangleq \max_{ \mathbf{d} \in [N]^{K} }  \max_{k \in [K] } \  \Pr  \left\{ (\hat{F}_{d_{k}} , \hat{W}_{k})  \neq (F_{d_{k}} , W_{k})  \right\}  
\end{equation}
which accounts for the worst-case file demand tuple amongst  all $N^{K}$ possible user demands.

It is instructive to work with the reciprocal of the content rate $R_{F}$, 
which enjoys desirable analytical properties, see, e.g., \cite{Maddah-Ali2015,Hachem2018}.
To this end, we define 
\begin{equation}
\mathcal{T} \triangleq \frac{1} { R_{F} } = \frac{T}{B}
\end{equation}
which is the number of physical channel uses required to communicate one bit of content to each user.
Since channel uses often correspond to time instance, $\mathcal{T} $ is referred to as the  \emph{delivery time} or \emph{delay}, used interchangeably.
Given a memory size $M$, a delay-rate trade-off is denoted by the tuple $(\mathcal{T},\mathbf{R} ; M)$,
which  is achievable if there exists a sequence of 
$(T,1/\mathcal{T},\mathbf{R},M)$ codes such that $P_{e,T}  \rightarrow 0$ as $T  \rightarrow \infty$.
For any $(\mathbf{R};M)$, the optimal (content) delivery time is defined as:\footnote{When describing a trade-off of performance measures (e.g. $\mathcal{C}(\mathcal{T} ; M)$), a semicolon separates performance measure arguments (e.g. $\mathcal{T}$) from arguments representing fixed system parameters (e.g. $M$).}
\begin{equation}
\label{eq:delay_rate_trade_off}
\mathcal{T}^{\star}(\mathbf{R} ; M) \triangleq \inf \big\{ \mathcal{T} :  (\mathcal{T},\mathbf{R} ; M) \ \text{is achievable}\big\}.
\end{equation}
Conversely, for any  $(\mathcal{T} ; M)$, the (non-content) capacity region is defined  as:
\begin{equation}
\label{eq:capacity_delay_trade_off}
\mathcal{C}(\mathcal{T} ; M) \triangleq \mathrm{cl} \big\{ \mathbf{R} : (\mathcal{T},\mathbf{R} ; M) \ 
\text{is achievable}  \big\}. 
\end{equation}
\begin{remark}
\label{remark:worst_case}
\textbf{(Worst-case demands).} 
We adopt a worst-case definition of 
performance measures (e.g. delay and capacity region) with respect to user demands---see the decoding error probability in \eqref{eq:error_prob}. 
Therefore, without loss of generality, we assume henceforth that demand tuples $\mathbf{d}$ comprise  $\min \{K,N \} $ distinct user demands.
Moreover, in scenarios where $ N < K$, worst-case demands occur when the first (i.e. weakest)
$N$ users make distinct file demands, as we will see further on in Section \ref{sec:achievability}.
A similar observation regarding the form of worst-case demand tuples when $N < K$ was made in \cite{Amiri2018b}, where the focus is on minimizing the transmit power subject to a constraint on $R_{F}$ (or $\mathcal{T}$) in the cache-aided degraded GBC,  in the absence of non-content messages.
\end{remark}
\subsection{GDoF and GNDT}
\label{subsec:GDoF_GNDT}
In defining the GDoF and GNDT limits, the dependency of the rates and delivery time on $P$ is highlighted.
That is, for any given $M$ and $P$, an achievable delay-rate tuple is denoted by $\big( \mathcal{T}(P),\mathbf{R}(P) ; M \big)$,
while $\mathcal{T}^{\star}(\mathbf{R}; M, P)$ and $\mathcal{C}(\mathcal{T} ; M,P)$ describe the optimal trade-offs.

We denote a GDoF tuple by $\mathbf{r} \triangleq (r_{1},\ldots, r_{K})$, where $r_{k}$ is the GDoF of user $k$,
while the GNDT is denoted by $\tau$.
For given $M$, a GNDT-GDoF trade-off $(\tau,\mathbf{r} ; M)$  is achievable if there exists 
a sequence of achievable delay-rate tuples $\big( \mathcal{T}(P),\mathbf{R}(P) ; M \big)$, for all $P$, such that
\begin{align}
r_{k} & = \lim_{P \rightarrow \infty} \frac{R_{k}(P)}{\log(1+P)}, \ \forall k \in [K] \\
\tau & =  \lim_{P \rightarrow \infty} \mathcal{T}(P) \log(1+P).
\end{align}
For any $(\mathbf{r} ; M)$, the optimal GNDT is defined as
\begin{equation}
\tau^{\star}(\mathbf{r}; M) \triangleq \inf \big\{ \tau :  (\tau,\mathbf{r} ; M) \ \text{is achievable}\big\}.
\end{equation}
On other hand, for any pair $(\tau ; M)$, the optimal GDoF region is defined as: 
\begin{equation}
\mathcal{D}(\tau ; M ) = \mathrm{cl}\big\{ \mathbf{r} : (\tau,\mathbf{r} ; M) \ 
\text{is achievable}  \big\}. 
\end{equation}
\begin{remark} 
\label{rem:time_slots}
\textbf{(Time slots).} 
We measure the GNDT in \emph{time slots}, where $1$ time slot corresponds to the overall delay of delivering a single file to the strongest user 
(i.e. user $K$) in the absence of caches, interference and instantaneous messages, as $P$ approaches infinity.
In this isolated single-user scenario, the delivery time (per bit) is given by $\mathcal{T}_{0}(P) = 1 / \log(1 + P)$, and the overall delay is given by $B \mathcal{T}_{0}(P) = B / \log(1 + P)$  in channel uses.
Now suppose that in a general setting with arbitrary number of users and cache sizes, we deliver a file of size $B$ to each user with delay $\mathcal{T}(P)$.
The corresponding GNDT is given by
\begin{equation}
\label{eq:GNDT_remark}
\tau \triangleq  \lim_{P \rightarrow \infty}  \mathcal{T}(P) \log(1+P) = \lim_{P \rightarrow \infty} \frac{B\mathcal{T}(P)}{B\mathcal{T}_{0}(P)}.
\end{equation}
The ratio in \eqref{eq:GNDT_remark} makes the definition of the GNDT and its time slot unit all the more clear.
\end{remark}
\begin{remark}
The delay, capacity, GNDT and GDoF characterizations we obtain in this work all depend on the normalized memory size $\mu$ instead of the actual memory size $M$.
This is reflected in the arguments  of the performance measures in the following sections, where $\mu$ replaces $M$.
Moreover, we highlight the dependency on the channel strength levels, e.g. $\tau^{\star}(\mathbf{r}; \mu, \bm{\alpha}) $
and  $\mathcal{D}(  \tau ; \mu, \bm{\alpha}) $,
and the nominal power parameter, e.g.   $\mathcal{C}( \mathcal{T}; \mu,  \bm{\alpha},P) $ and 
$\mathcal{T}^{\star}(\mathbf{R} ; \mu,  \bm{\alpha},P)$.
\end{remark}
\section{Main Result and Insights}
\label{sec:main_results}
We start this section by defining an  upper bound for the GNDT given any GDoF tuple $\mathbf{r}$.
\begin{definition}
\label{def:GNDT_ub}
For any $\mu$,  $\bm{\alpha}$ and $\mathbf{r}$, where the GDoF tuple $\mathbf{r}$ is feasible with components satisfying $\sum_{i \in [k]}  r_{i} \leq \alpha_{k}$ for all $k \in [K]$, we define\footnote{In \eqref{eq:tau_ub}, and throughout this work, we use the convention $\binom{n}{k} = 0$, for all $n < k$.} 
\begin{equation}
\label{eq:tau_ub}
\tau^{\mathrm{ub}}(\mathbf{r} ; \mu,  \bm{\alpha}) \triangleq   \max_{k \in [K]} \left\{ 
\frac{ 1}{ \big( \alpha_{k} - \sum_{i \in [k]} r_{i} \big) }  \cdot  \mathrm{conv} \left(
\frac{ \binom{K}{K \mu +1} - \binom{K - \min\{k,N\}}{K \mu + 1}  }{  \binom{K}{K \mu}  }   \right) \right\}
\end{equation}
where $ \mathrm{conv}  \big( f(K \mu) \big)$ denotes the lower convex envelope of the 
points $\left\{ \big( K \mu , f( K \mu) \big) : K \mu \in [0:K] \right\}$.
\end{definition}
Equipped with Definition \ref{def:GNDT_ub}, we are now ready to present the main theorem of this work.
\begin{theorem}
\label{theorem:optimal_trade_off}
The GNDT-GDoF trade-off described by $\tau^{\mathrm{ub}}(\mathbf{r} ; \mu,  \bm{\alpha}) $  in \eqref{eq:tau_ub} is achievable. 
Moreover, $\tau^{\mathrm{ub}}(\mathbf{r} ; \mu,  \bm{\alpha}) $ is within a multiplicative factor of $2.01$ from the optimal trade-off,
that is\footnote{The multiplicative factor is in fact $2.00884$, which is consistent with the result of Yu et al. \cite{Yu2019}.}
\begin{equation}
\label{eq:order_optimality}
\frac{1}{2.01} \cdot \tau^{\mathrm{ub}}(\mathbf{r} ; \mu,  \bm{\alpha})  \leq \tau^{\star}(\mathbf{r} ; \mu,  \bm{\alpha})  \leq \tau^{\mathrm{ub}}(\mathbf{r} ; \mu,  \bm{\alpha}) .
\end{equation}
\end{theorem}
The achievability of Theorem \ref{theorem:optimal_trade_off} is presented in Sections \ref{sec:degraded_GBC_unicast_multicast} and \ref{sec:achievability}, with some details relegated to Appendix \ref{appendix:GDoF_region}. 
For ease of exposition, the focus of these sections is on integer values of $K \mu$, while 
the extension to non-integer $K \mu$ is relegated to Appendix \ref{appendix:non_integer_K_mu}.
On the other hand, the converse of Theorem \ref{theorem:optimal_trade_off}  is presented in Section \ref{sec:converse}.
Next, we draw some insights from the main result. We start by focusing on integer $K \mu$,
and then discuss the case with non-integer $K \mu$ further on.
\subsection{Separation principle}
\label{subsec:separation}
The achievability of $\tau^{\mathrm{ub}}(\mathbf{r} ; \mu,  \bm{\alpha}) $ employs a separation-based strategy, which isolates the
content caching and delivery problem from the physical-layer transmission problem  \cite{Hachem2018}.
In particular, caching, generating coded multicast messages (XORs), and recovering demanded files from received multicast messages and local cache contents are all carried out at the bit level in the noiseless shared link manner \cite{Maddah-Ali2014,Yu2018,Yu2019}, 
and are oblivious to the transmission strategy over the physical channel.
On the other hand, the physical channel sees $\binom{K}{K\mu + 1}$ multicast messages (coded content) and $K$ unicast messages (non-content),
and communicates them in a joint multicast and unicast fashion. 

The physical-layer scheme employs power control with superposition coding and successive decoding.
Hence different GNDT-GDoF trade-offs, described by the relationship in \eqref{eq:tau_ub},  are achieved by tuning the
underlying power allocation and GDoF assignment problems. A detailed exposition of the physical-layer scheme is given in Section \ref{sec:degraded_GBC_unicast_multicast} (see also Appendix \ref{appendix:GDoF_region}).
\subsection{GNDT in the absence of non-content messages}
As a special case of Theorem \ref{theorem:optimal_trade_off},
we recover the achievability result in \cite{Lampiris2019}, where it was shown that 
for $N \geq K$ and integer $K \mu$, and in the absence of non-content messages, one can achieve
\begin{equation}
\label{eq:GNDT_no_unicast}
\tau^{\mathrm{ub}}(\mathbf{0} ; \mu,  \bm{\alpha})  = \max_{k \in [K]} \left\{ 
\frac{ 1}{ \alpha_{k}  }  \cdot
\frac{ \binom{K}{K \mu +1} - \binom{K - k}{K \mu + 1}  }{  \binom{K}{K \mu}  }  \right\}.
\end{equation}
The order-optimality of $\tau^{\mathrm{ub}}(\mathbf{0} ; \mu,  \bm{\alpha}) $ up to a multiplicative factor of $4.02$ is also proved in \cite{Lampiris2019},
which we tighten in Theorem \ref{theorem:optimal_trade_off} by a factor of $2$.
Moreover, in addition to strengthening the order-optimality result, our new achievability proof (given in Sections \ref{sec:degraded_GBC_unicast_multicast} and \ref{sec:achievability}) sheds new light on the achievable GNDT in \eqref{eq:GNDT_no_unicast}, and provides an operational interpretation in terms of the 
multiple multicast GDoF region of the underlying degraded GBC.

To illustrate, consider a setting with $K = N= 3$ and $M = 1$, and assume that each user requests a distinct file.
Employing a separation-based strategy, a standard coded caching scheme delivers $3$ coded multicast messages: $W_{12}$, $W_{13}$ and $W_{23}$, each
designated to a pair of users specified by the message index; and each of size 
$1 / 3$ in (normalized) file units.
On the other hand, the physical channel communicates the coded messages in a multiple multicast fashion and, as shown in Theorem \ref{theorem:GDoF_phy} 
in Section \ref{sec:degraded_GBC_unicast_multicast}, operates at any non-negative GDoF 
tuple $(r_{12},r_{13},r_{23})$ that satisfies\footnote{Note that the multicast GDoF tuple 
here should not be confused with the non-content GDoF tuple $\mathbf{r}$. This will be further clarified in Section \ref{sec:degraded_GBC_unicast_multicast}.}
\begin{equation}
\label{eq:multicast_GDoF_region_3_user}
\begin{aligned}
r_{12} + r_{13} & \leq \alpha_{1} \\
r_{12} + r_{13} + r_{23} & \leq \alpha_{2}.
\end{aligned}
\end{equation}
The GDoF region in \eqref{eq:multicast_GDoF_region_3_user} admits an intuitive interpretation. User $1$ recovers both $W_{12}$ and $W_{13}$, and hence the sum-GDoF of these messages is bounded by the channel strength $\alpha_{1}$.
Due to the degradedness of the physical channel, user $2$ can decode whatever user $1$ decodes, and must additionally recover $W_{23}$. Therefore, the total GDoF 
cannot exceed $\alpha_{2}$.
It is clear that user $3$, i.e. the strongest user, can recover all messages as 
$\alpha_{2} \leq \alpha_{3}$.  
Since all $3$ coded messages are of equal size, it is most efficient to operate at the symmetric multicast GDoF  $r_{\mathrm{sym}} = \min \left\{ \frac{\alpha_{1}}{2} , \frac{\alpha_{2}}{3} \right\}$, which is directly computed from \eqref{eq:multicast_GDoF_region_3_user}---see Fig. \ref{fig:signal_levels}. 
The achievable GNDT is hence given by
\begin{equation}
\tau = \frac{1}{3} \cdot \frac{1}{r_{\mathrm{sym}}}
\end{equation}
which exactly coincides with \eqref{eq:GNDT_no_unicast} for $K = 3$ and $\mu = 1/3$.
\begin{figure}[t]
\centering
\includegraphics[width = 0.8\textwidth]{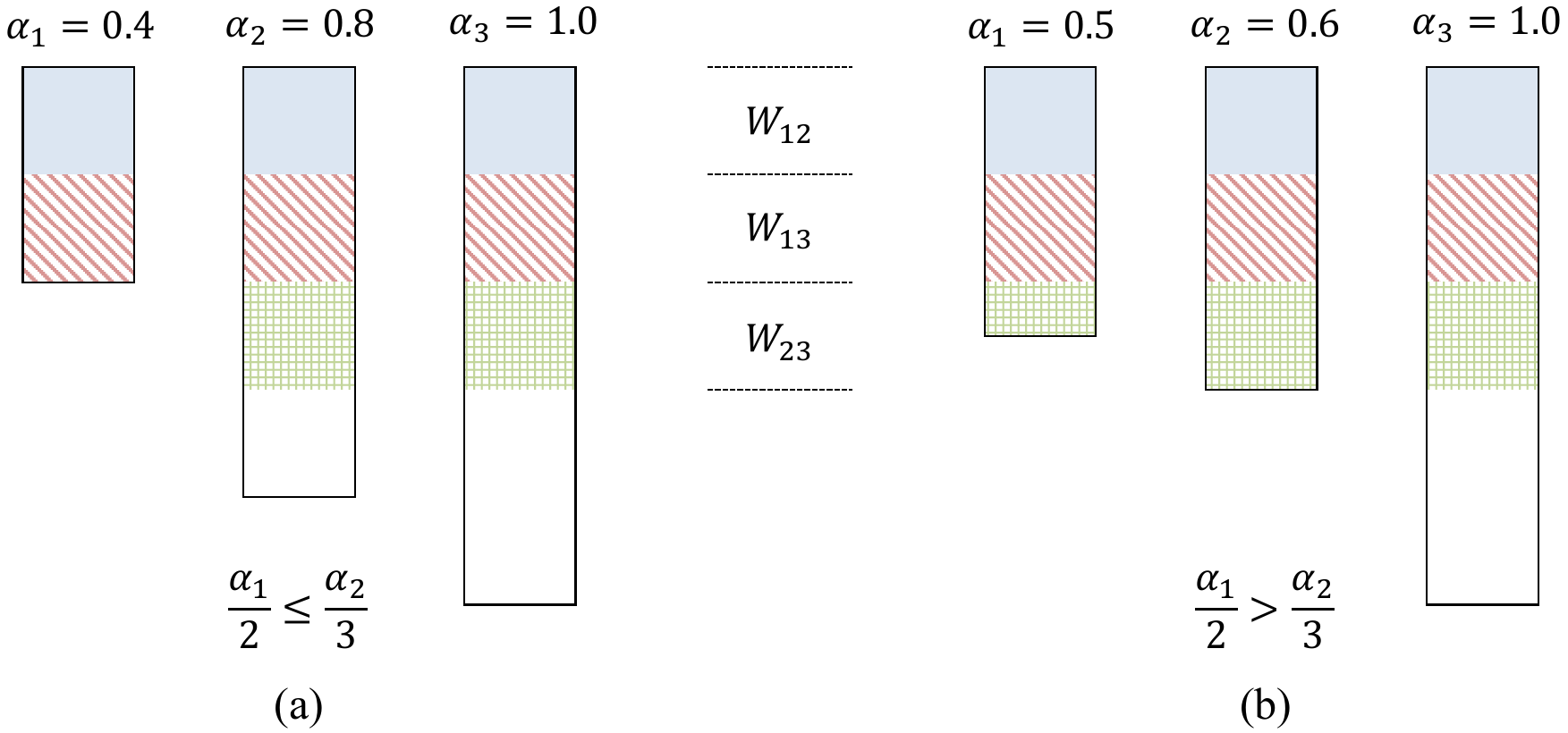}
\caption{\small Received signal power levels in a $3$-user degraded GBC with multiple multicast messages, where $W_{ij}$ is intended to users $i$ and $j$ (group size of $2$), and a symmetric GDoF of $r_{\mathrm{sym}} = 0.2$ is achieved in both (a) and (b). Top levels represent signals transmitted with higher powers, received by all users above their respective noise levels 
(bottom end of each bar). Bottom levels represent signals transmitted with lower powers, heard by sufficiently strong users and corrupted by noise 
(hence clipped) at weaker users. Multicast messages (coloured levels) can carry coded content (e.g. $\mu = 1 / 3$).
Uncoloured signal levels (in white) are unoccupied, representing \emph{topological holes} for communicating non-content messages (see Corollary \ref{corollary:GDoF_region_min_GNDR}).}
\label{fig:signal_levels}
\end{figure}
As it turns out, the same argument extends to general settings, 
where the achievable GNDT in \eqref{eq:GNDT_no_unicast} is decomposed as 
\begin{equation}
\label{eq:GNDT_no_unicast_sep}
\tau^{\mathrm{ub}}(\mathbf{0} ; \mu,  \bm{\alpha}) = \underbrace{ \frac{ 1  }{  \binom{K}{K \mu}  } }_{\text{norm. size}}  \cdot \underbrace{  \max_{k \in [K - K \mu]}  \left\{ 
\frac{ \binom{K}{K \mu +1} - \binom{K - k}{K \mu + 1}  }{ \alpha_{k}  }   \right\}}_{1 / r_{\mathrm{sym}}}.
\end{equation}
The first term on the right-hand-side of \eqref{eq:GNDT_no_unicast_sep} is the normalized size of each coded multicast message; while the second term is the reciprocal of the symmetric multiple multicast GDoF, derived from the multiple multicast GDoF region of the underlying GBC 
(see Corollary \ref{corollary:symm_multicast_region}, Section \ref{subsec:symm_multicast_subset}).
\subsection{Achievable GDoF under minimum GNDT}
Let us now plug non-content messages back in, while maintaining the assumption that $N \geq K$ for ease of exposition.  Theorem \ref{theorem:optimal_trade_off} suggests that in scenarios with asymmetric channel strengths, 
the order-optimal GNDT in \eqref{eq:GNDT_no_unicast}, achieved by eliminating non-content messages, can be maintained while simultaneously achieving non-zero GDoF for (some) non-content messages.
To see this, let us define user $k^{\star}$ (\emph{bottleneck user} in \cite{Lampiris2019}) such that
\begin{equation}
k^{\star} \triangleq  \arg \max_{k \in [K]} \left\{
\frac{ \binom{K}{K \mu +1} - \binom{K - k}{K \mu + 1}  }{ \alpha_{k}  }  \right\}
\end{equation}
For example, in the illustrations shown in Fig. \ref{fig:signal_levels}, with $K =3$ and $K \mu + 1  = 2$, we have  $k^{\star}  = 1$  in (a), where $\frac{\alpha_{1}}{2} \leq \frac{\alpha_{2}}{3}$; and $k^{\star}  = 2$ in (b), where $\frac{\alpha_{1}}{2} > \frac{\alpha_{2}}{3}$.
In terms of the multiple multicast GDoF region of the underlying GBC, $k^{\star}$ is the 
(smallest) index such that
the inequality that delimits the sum-GDoF of messages decoded by user
$k^{\star}$ holds with equality (see Theorem \ref{theorem:GDoF_phy}).

From \eqref{eq:tau_ub}, it follows that 
achieving a GNDT of  $\tau^{\mathrm{ub}}(\mathbf{0} ; \mu,  \bm{\alpha}) $ through the proposed strategy  requires setting 
$r_{k} = 0$, for all $k \in [k^{\star}]$.
That is, we cannot send additional information to the bottleneck user $k^{\star}$, or weaker users whose messages are also decodable by user $k^{\star}$, without increasing the achievable GNDT.
However, users in $[k^{\star} + 1: K] $ can achieve non-zero non-content GDoF without affecting the GNDT in \eqref{eq:GNDT_no_unicast}, by communicating through the \emph{topological holes} arising from the asymmetry in channel strength levels, specifically when 
 $\alpha_{k^{\star}} < \alpha_{k^{\star}+1}$.
These achievable non-content GDoF tuples are described as follows.
\begin{corollary}\textbf{(Topological Holes).}
\label{corollary:GDoF_region_min_GNDR}
A minimum GNDT of  $\tau^{\mathrm{ub}}(\mathbf{r} ; \mu,  \bm{\alpha}) = \tau^{\mathrm{ub}}(\mathbf{0} ; \mu,  \bm{\alpha}) $  and a non-content GDoF tuple $\mathbf{r}$  are simultaneously achievable given that the GDoF tuple $\mathbf{r} \in \mathbb{R}_{+}^{K}$  satisfies 
\begin{equation}
\label{eq:GDoF_inequalities_min_tau}
\begin{aligned}
r_{k} & = 0, \ \forall k \in [k^{\star}] \\
r_{k^{\star} + 1} + \cdots + r_{k}  & \leq  \alpha_{k^{\star} + 1} - \alpha_{k^{\star}} \cdot 
\left(
\frac{ \binom{K}{K \mu +1} - \binom{K - k}{K \mu + 1}  }{ \binom{K}{K \mu +1} - \binom{K - k^{\star}}{K \mu + 1}  } \right) , \ \forall k \in [k^{\star}+1 : K].
\end{aligned} 
\end{equation}
\end{corollary}
Examples that illustrate Corollary \ref{corollary:GDoF_region_min_GNDR} using signal power levels, measured in terms of the exponent of $P$ (see, e.g., \cite{Bresler2010,Jafar2010}), are shown in Fig. \ref{fig:signal_levels}.
As argued in Section \ref{subsec:intro_setting}, the current treatment of content traffic and non-content traffic as two independent entities leads 
to scheduling the two types of traffic on orthogonal wireless resource blocks, which is suboptimal in general.
This observation is made concrete in the following remark by leveraging Corollary \ref{corollary:GDoF_region_min_GNDR}. 
\begin{remark}
Suppose that we wish to deliver content at the minimum achievable GNDT given by 
$ \tau = \tau^{\mathrm{ub}}(\mathbf{0} ; \mu,  \bm{\alpha})$.
From Corollary \ref{corollary:GDoF_region_min_GNDR}, we know that we can simultaneously communicate a non-content message to, e.g., 
user $K$ with a GDoF of $r_{K}$, which satisfies the corresponding inequality in \eqref{eq:GDoF_inequalities_min_tau}; hence delivering  
additional non-content information of $ \tau \cdot r_{K}$ in normalized file units.\footnote{Recall from Remark \ref{rem:time_slots} that $\tau$, in times slots, is measured per file delivered over a channel with a GDoF of $1$. Therefore,  $ \tau \cdot r_{K}$ corresponds to delivered information in file units.}
An alternative approach is to deliver content traffic and non-content traffic over orthogonal resource blocks using, e.g., time-sharing (see Section \ref{subsec:intro_setting}).
This incurs an additional delay of at least $\tau \cdot r_{K} / \alpha_{K}$ time slots, required to deliver the same amount of non-content information separately.
For the examples shown in  Fig. \ref{fig:signal_levels}, this corresponds to an increase of $40 \%$ in communication delay.
\end{remark}
\subsection{Non-integer $K \mu$}
\label{subsec:non_integer_K_mu}
Using the separation-based strategy described above,
a GNDT-GDoF trade-off of $ \tau^{\mathrm{ub}}(\mathbf{r} ; \mu,  \bm{\alpha}) $  is achieved for all $\mu$ such that  $K \mu$ is an integer. In this case, the $\mathrm{conv}(\cdot)$ operator in \eqref{eq:tau_ub}  is dropped and the corresponding achievable GNDT can be expressed by
\begin{equation}
\label{eq:tau_ub_int}
\tau^{\mathrm{ub}}_{K \mu \in \mathbb{Z}} (\mathbf{r} ; \mu,  \bm{\alpha})  \triangleq  \max_{k \in [K]} \left\{ 
\frac{ 1}{ \big( \alpha_{k} - \sum_{i \in [k]} r_{i} \big) }  \cdot  
\frac{ \binom{K}{K \mu +1} - \binom{K - \min\{k,N \} }{K \mu + 1}  }{  \binom{K}{K \mu}  }  \right\}.
\end{equation}
For $\mu$ such that $K \mu$ takes non-integer values drawn from $(0,K)$, 
a  standard memory-sharing argument \cite{Maddah-Ali2014} achieves the lower convex envelope of 
the points in \eqref{eq:tau_ub_int}, defined as
\begin{equation}
\label{eq:tau_ub_conv}
\tau_{\mathrm{ms}}^{\mathrm{ub}} (\mathbf{r} ; \mu,  \bm{\alpha})  \triangleq \mathrm{conv} \left(
\max_{k \in [K]} \left\{ 
\frac{ 1}{ \big( \alpha_{k} - \sum_{i \in [k]} r_{i} \big) }  \cdot 
\frac{ \binom{K}{K \mu +1} - \binom{K -  \min\{k,N \} }{K \mu + 1}  }{  \binom{K}{K \mu}  }  \right\}
 \right).
\end{equation}
In this straightforward application of the memory-sharing principle, files, caches and transmissions are divided proportionally such that the system effectively operates as two systems: one with a multicasting gain of $\lfloor K \mu + 1 \rfloor$, and another with a multicasting gain of 
$\lceil K \mu + 1 \rceil$.
Nevertheless, it turns out that this strategy can be strictly improved upon, especially in asymmetric settings.

In the improved strategy, caching and preparing the sets of coded multicast messages are carried out as in the standard
memory-sharing scheme.
Nevertheless, instead of carrying out the physical-layer transmission \emph{sequentially} in two phases, 
the degraded GBC \emph{jointly} delivers two sets of coded multicast messages,
one with messages intended to $\lfloor K \mu + 1 \rfloor$ users each and another with messages intended  to $\lceil K \mu + 1 \rceil$ users each; as well as the non-content unicast message set.\footnote{The scheme in \cite{Amiri2018b} (implicitly) adopts a similar superposition strategy with two sets of multicast messages.}
This joint delivery strategy achieves the GNDT in  \eqref{eq:tau_ub},
which satisfies 
\begin{equation}
\label{eq:tau_tilde_tau_breve}
\tau^{\mathrm{ub}} (\mathbf{r} ; \mu,  \bm{\alpha})  \leq  \tau_{\mathrm{ms}}^{\mathrm{ub}} (\mathbf{r} ; \mu,  \bm{\alpha}) .
\end{equation}
The inequality in \eqref{eq:tau_tilde_tau_breve} is strict in asymmetric scenarios at some values of $\mu$, as seen through the example in Fig. \ref{fig:GNDT_memory}.  
Details and derivations related to this part  can be found in Appendix \ref{appendix:non_integer_K_mu}.
\begin{figure}[h]
\centering
\includegraphics[width = 0.55\textwidth]{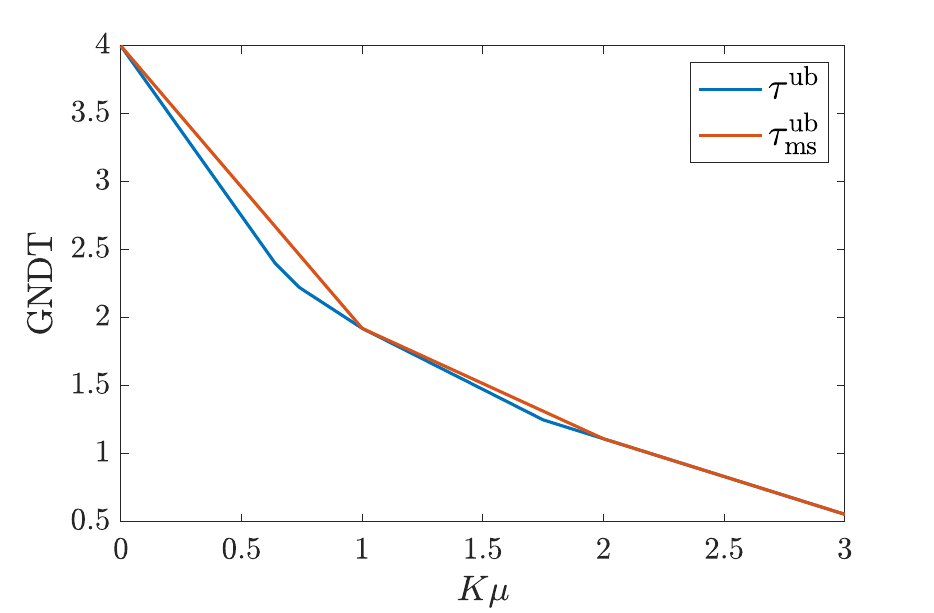}
\caption{\small GNDT-memory trade-off curves for $K = N = 4$,  $\bm{\alpha} = (0.45,0.65,0.85,1)$ and $\mathbf{r} = \mathbf{0}$.}
\label{fig:GNDT_memory}
\end{figure}
\begin{remark}
The inequality in \eqref{eq:tau_tilde_tau_breve} is another manifestation of the fact that in the degraded GBC, 
superposition coding is in general superior to time-sharing.
A naive application of the memory-sharing principle leads to delivering the 
two sets of coded multicast messages sequentially in a time shared fashion. 
While this incurs no loss in symmetric settings (as in, e.g., \cite{Maddah-Ali2014}), it can be strictly suboptimal in 
non-symmetric settings.
On the other hand, the scheme described in  Appendix \ref{appendix:non_integer_K_mu}
takes advantage of asymmetry in the degraded GBC 
through superposition coding. 
\end{remark}
\subsection{GDoF region and delay-rate trade-off}
Theorem \ref{theorem:optimal_trade_off} leads to a characterization of the GDoF region $\mathcal{D} ( \tau ; \mu,\bm{\alpha}) $, given as follows. 
\begin{corollary}
\label{corollary:GDoF_region}
For any  $\mu$,  $\bm{\alpha}$ and $\tau$, the GDoF region $\mathcal{D} ( \tau ; \mu,\bm{\alpha})$ satisfies:
\begin{equation}
\label{eq:GDoF_region_bounds}
\mathcal{D}_{\mathrm{in}} ( \tau ; \mu,\bm{\alpha}) \subseteq \mathcal{D} ( \tau ; \mu,\bm{\alpha}) \subseteq \mathcal{D}_{\mathrm{in}}(2.01 \cdot \tau ; \mu , \bm{\alpha})
\end{equation}
where $\mathcal{D}_{\mathrm{in}}( \tau ; \mu,\bm{\alpha}) $ is the set of 
all non-content GDoF tuples $\mathbf{r} \in \mathbb{R}_{+}^{K}$ satisfying
\begin{equation}
\sum_{i \in [k]} r_{i}   +  \frac{1}{\tau} \cdot \mathrm{conv} \left( \frac{ \binom{K}{K \mu +1} - \binom{K - \min\{k,N\} }{K \mu + 1}  }{ \binom{K}{K \mu} } \right)
\leq \alpha_{k}, \ \forall k\in [K].
\end{equation}
\end{corollary}
The achievable GDoF region in  Corollary \ref{corollary:GDoF_region_min_GNDR} is a special case of the one in Corollary \ref{corollary:GDoF_region}---the former is recovered by setting $\tau = \tau^{\mathrm{ub}}(\mathbf{0} ; \mu,  \bm{\alpha}) $ in $\mathcal{D}_{\mathrm{in}}( \tau ; \mu,\bm{\alpha}) $, while restricting to 
$N \geq K$. We conclude this section with the following remark on  characterizing  the optimal delay-rate trade-off.
\begin{remark}
\label{remark:delay_rate_char}
\textbf{(Approximate delay-rate characterization).} 
As one would hope, the GNDT-GDoF-based characterizations presented in this section translate to counterpart
approximate delay-rate characterizations.
This is shown in Appendix \ref{appendix:constant_gap}, where we characterize the set of all achievable delay-rate trade-off tuples 
$(\mathcal{T}, \mathbf{R} ; \mu)$ up to an additive gap of $2$ bits per channel use for rates and a multiplicative gap of $2.01$ for the delay, irrespective of all system parameters.
\end{remark}
\section{Degraded GBC with Unicast and Multiple Multicast Messages}
\label{sec:degraded_GBC_unicast_multicast}
In this section, we focus on a variant of the degraded  GBC in Section \ref{subsec:phy_channel} with no caches and with two message sets: a unicast message set and a multiple multicast message set.
The latter message set is referred to as the $\sigma$-multicast message set, where $\sigma \in [2:K]$ is the size of the corresponding multicast groups.\footnote{The case with $\sigma = 1$ is ignored as it reduces to having only a unicast message set.}
As seen in the following section, this channel model is at the heart of the separation architecture---unicast messages carry instantaneous non-content traffic (i.e. messages), while $\sigma$-multicast messages carry coded content traffic (i.e. files). 
It is worthwhile highlighting that the size of multicast groups $\sigma$ remains fixed once selected.
\subsection{Unicast and $\sigma$-multicast message sets}
\label{subsec:unicast_sigma_multicast_messages}
The unicast message set is given by $\left\{W_{k} : k \in [K] \right\}$, where each message $W_{k}$ is intended to the corresponding user $k$ 
and has a rate of $R_{k}$ and a GDoF of $r_{k}$; while  $\sigma$-multicast message set is given by 
$\left\{W_{\mathcal{S}} : \mathcal{S} \subseteq [K],  |\mathcal{S}| = \sigma \right\}$, with each message 
$W_{\mathcal{S}}$ intended to all users in $\mathcal{S}$ and has a rate of $R_{\mathcal{S}}$ and a GDoF of $r_{\mathcal{S}}$.
Note that since $\sigma \geq 2$, there is no ambiguity between $W_{k} $ and $W_{\mathcal{S}} $, 
$R_{k} $ and $R_{\mathcal{S}} $, or $r_{k} $ and $r_{\mathcal{S}} $, for any $k \in [K]$ and $\mathcal{S} \subseteq [K]$.

For any $\sigma$, $\bm{\alpha}$ and $P$, the capacity region and GDoF region of the above channel are denoted by
 $\mathcal{C}^{\mathrm{PHY}}(\sigma,\bm{\alpha},P)$  and $\mathcal{D}^{\mathrm{PHY}}(\sigma,\bm{\alpha})$, respectively.
We  define the set of all $\sigma$-multicast groups  as $\Sigma \triangleq \left\{ \mathcal{S} \subseteq [K] :  |\mathcal{S}| = \sigma \right\}$,
where $|\Sigma | = \binom{K}{\sigma} $.
Moreover, we introduce a family of subsets of $\Sigma$ given by  $\{ \Sigma_{i} : i \in [K - \sigma +1]  \}$,
where each member $\Sigma_{i} \subseteq \Sigma$ is defined as:\footnote{As an example, take $K = 4$ and $\sigma = 2$. Here we have $\Sigma = \big\{ \{1,2\}, \{1,3\}, \{1,4 \}, \{ 2,3\}, \{2,4\}, \{ 3,4\} \big\}$, which is partitioned into $\Sigma_1 = \big\{ \{1,2\}, \{1,3\}, \{1,4 \}\big\}$, 
$\Sigma_2 = \big\{ \{2,3\}, \{2,4\}\big\}$ and $\Sigma_3 = \big\{ \{3,4\}\big\}$. }
\begin{equation}
\label{eq:Sigma_i}
\Sigma_{i} \triangleq  \left\{ \mathcal{S} \in \Sigma :  \min\{ \mathcal{S} \} = i \right\}.
\end{equation}
It can be verified that $\{ \Sigma_{i} : i \in [K - \sigma +1]  \}$ is a partition of $\Sigma$, that is:
\begin{equation}
\bigcup_{i \in [K - \sigma + 1]} \Sigma_{i}  = \Sigma \ \ \text{and} \ \ \Sigma_{i} \cap \Sigma_{j} = \emptyset, \ \forall  i \neq j.
\end{equation}
We are now ready to present a characterization of the GDoF region $\mathcal{D}^{\mathrm{PHY}}(\sigma,\bm{\alpha})$.
\begin{theorem}
\label{theorem:GDoF_phy}
For the above described degraded GBC with unicast and $\sigma$-multicast messages, the GDoF region 
$\mathcal{D}^{\mathrm{PHY}}(\sigma,\bm{\alpha})$ is given by all tuples 
$(r_{k}: k\in [K], \ r_{\mathcal{S}} : \mathcal{S} \in \Sigma) \in \mathbb{R}_{+}^{K + \binom{K}{\sigma}}$ that satisfy
\begin{equation}
\label{eq:GDoF_region_phy}
\begin{aligned}
\sum_{i \in [k]} r_{i} + \sum_{\mathcal{S} \in \cup_{i \in [k]} \Sigma_{i} }  r_{\mathcal{S}}  & \leq \alpha_{k}, \ \forall k \in [K - \sigma + 1]\\
\sum_{i \in [k]}  r_{i} + \sum_{\mathcal{S} \in \Sigma}  r_{\mathcal{S}}  & \leq \alpha_{k}, \ \forall k \in [K - \sigma + 2 : K].
\end{aligned}
\end{equation}
\end{theorem}
The GDoF region in Theorem \ref{theorem:GDoF_phy} is achieved using a scheme based on power control with superposition coding and  successive decoding.
The full proof is relegated to Appendix \ref{appendix:GDoF_region}. 
Theorem \ref{theorem:GDoF_phy} has an intuitive interpretation, which is best seen by laying out the inequalities in \eqref{eq:GDoF_region_phy} as
\begin{align}
\label{eq:GDoF_region_phy_2_1}
r_{1} + \sum_{\mathcal{S} \in \Sigma_{1}} r_{\mathcal{S}}  & \leq \alpha_{1} \\
\label{eq:GDoF_region_phy_2_2}
r_{1} + r_{2} + \sum_{\mathcal{S} \in \Sigma_{1} \cup \Sigma_{2}   } r_{\mathcal{S}}  & \leq \alpha_{2} \\
\nonumber
& \vdots \\
\label{eq:GDoF_region_phy_2_3}
\sum_{i \in [K- \sigma + 1] } r_{i}  + \sum_{\mathcal{S} \in \Sigma   } r_{\mathcal{S}}    & \leq \alpha_{K- \sigma + 1} \\
\nonumber
& \vdots \\
\label{eq:GDoF_region_phy_2_4}
\sum_{i \in [K] } r_{i}  + \sum_{\mathcal{S} \in \Sigma   }  r_{\mathcal{S}}   & \leq \alpha_{K}.
\end{align}
User $1$ recovers all messages in $\{W_{1}, W_{\mathcal{S}} : \mathcal{S} \in \Sigma_{1} \}$, 
and hence the sum-GDoF of such messages cannot exceed the channel strength of this user, as seen in \eqref{eq:GDoF_region_phy_2_1}.
Due to the degradedness of the channel, user $2$ can recover whatever user $1$ recovers, and must also decode for messages in
$\{W_{2}, W_{\mathcal{S}} : \mathcal{S} \in \Sigma_{2} \}$. This bounds the sum-GDoF of messages in  $\{W_{1},W_{2}, W_{\mathcal{S}} : \mathcal{S} \in \Sigma_{1} \cup \Sigma_{2}  \}$ by the channel strength of user $2$, as seen in \eqref{eq:GDoF_region_phy_2_2}.
The same argument applies to all users up to user $K -\sigma + 1$,  as seen in \eqref{eq:GDoF_region_phy_2_3}.
Beyond user $K -\sigma + 1$, each user $k$ in $[K -\sigma + 2: K]$ is capable of recovering all messages in 
$\{W_{1},\ldots,W_{k-1}, W_{\mathcal{S}} : \mathcal{S} \in \Sigma \}$, and must additionally decode for message $W_{k}$.
This yields the sum-GDoF bounds in the second line of \eqref{eq:GDoF_region_phy} (see, e.g., \eqref{eq:GDoF_region_phy_2_4}).
\begin{remark}
\label{remark:Sigma_augment}
We augment the definition of the family of subsets given by  $\{ \Sigma_{i} : i \in [K - \sigma +1]  \}$ to include 
$\{ \Sigma_{i} : i \in [K - \sigma +2:K]  \}$, where we set $\Sigma_{i} = \emptyset$, for all $ i \in [K - \sigma +2 : K] $.
This allows us to express the inequalities in \eqref{eq:GDoF_region_phy} compactly as 
\begin{equation}
\label{eq:GDoF_region_phy_compact}
\begin{aligned}
\sum_{i \in [k]} r_{i} + \sum_{\mathcal{S} \in \cup_{i \in [k]} \Sigma_{i} }  r_{\mathcal{S}}   \leq \alpha_{k}, \ \forall k \in [K ].
\end{aligned}
\end{equation}
Moreover, throughout this paper, we use the convention $|\emptyset| = 0$.
\end{remark}
\begin{remark}
The characterization of $\mathcal{D}^{\mathrm{PHY}}(\sigma,\bm{\alpha})$ in Theorem \ref{theorem:GDoF_phy} leads to 
a characterization of the capacity region $\mathcal{C}^{\mathrm{PHY}}(\sigma,\bm{\alpha},P)$ up to a constant gap.
Details are relegated to Appendix \ref{appendix:subsec_constant_gap_PHY}.
While this constant gap result is of interest in its own right, its main significance to this work is that it lays the ground for 
establishing the approximate delay-rate characterization  in Appendix \ref{appendix:constant_gap}.
\end{remark}
\subsection{Symmetric $\sigma$-multicast GDoF}
\label{subsec:symm_multicast_subset}
We are interested in scenarios where in addition to  unicast messages,
we wish to communicate a subset of the  $\sigma$-multicast messages at a symmetric rate.
This is specified as follows.
\begin{itemize}
\item For a given parameter $s \in [K]$, we wish to communicate the subset of $\sigma$-multicast messages where each message is intended to at least one user in $[s]$.
\item For the communicated $\sigma$-multicast messages, we wish to achieve a symmetric GDoF of 
$r_{\mathrm{sym}}$. 
\end{itemize}
From \eqref{eq:Sigma_i}, it follows that for any $s \in [K]$, the set of $\sigma$-multicast messages
of interest is given by
\begin{equation}
\big\{W_{\mathcal{S}} : \mathcal{S} \in \Sigma,  \mathcal{S} \cap [s] \neq \emptyset  \big\} = 
\big\{ W_{\mathcal{S}} : \mathcal{S} \in \cup_{i \in [s]} \Sigma_{i} \big\}.
\end{equation}
It can be verified that the above set comprises all $\sigma$-multicast messages whenever $s \geq K - \sigma + 1$.
For this scenario of interest, we define a lower dimensional projection 
of $\mathcal{D}^{\mathrm{PHY}}(\sigma,\bm{\alpha})$ as:
\begin{multline}
\mathcal{D}^{\mathrm{PHY}}_{\mathrm{sym}}(\sigma,\bm{\alpha}, s) \triangleq 
\Big\{(r_{1},\ldots,r_{K},r_{\mathrm{sym}}) : (r_{k}: k\in [K], r_{\mathcal{S}} : \mathcal{S} \in \Sigma) \in \mathcal{D}^{\mathrm{PHY}}(\sigma,\bm{\alpha}), \\
 r_{\mathcal{S}} \geq r_{\mathrm{sym}} , \forall \mathcal{S} \in  \cup_{i \in [s]} \Sigma_{i}, 
\  \text{and} \ 
 r_{\mathcal{S}} = 0 , \forall \mathcal{S} \in  \cup_{i \in [s+1 : K]} \Sigma_{i}  \Big\}
\end{multline}
which is parametrized by $s$, in addition to $\sigma$ and $\bm{\alpha}$.
A characterization of $\mathcal{D}^{\mathrm{PHY}}_{\mathrm{sym}}(\sigma,\bm{\alpha}, s) $ is directly obtained from  Theorem \ref{theorem:GDoF_phy},
and is given by all tuples $(r_{k}: i \in [K], \ r_{\mathrm{sym}}) \in \mathbb{R}^{K+1}$ that satisfy:
\begin{equation}
\label{eq:sym_GDoF_region_phy}
\sum_{i \in [k]} r_{i} + \bigg| \bigcup_{i \in [ \min \{k,s\} ]} \Sigma_{i} \bigg| \cdot r_{\mathrm{sym}} \leq \alpha_{k}, \ \forall k \in [K].
\end{equation}
Next, we observe that the following identity holds
\begin{equation}
\label{eq:k_Sigma_size}
\bigg| \bigcup_{i \in [j]} \Sigma_{i} \bigg| = \sum_{i \in [j]}  |  \Sigma_{i} |  = \binom{K}{\sigma} - \binom{K - j}{\sigma}, \ \forall j \in [K].
\end{equation}
This is deduced by noting that  $|\Sigma | = \binom{K}{\sigma} $ and   $| \cup_{i \in [j+1:K]} \Sigma_{i} |  = \binom{K - j}{\sigma}$,
where the latter follows from the fact that $\cup_{i \in [j+1:K]} \Sigma_{i} $ is the family of all subsets of $[j+1:K]$ with size $\sigma$.
Since $| \cup_{i \in [j]} \Sigma_{i} | = | \Sigma|  - | \cup_{i \in [j+1:K]} \Sigma_{i} | $,  the identity in \eqref{eq:k_Sigma_size} holds.
By setting $j$ in \eqref{eq:k_Sigma_size}  to $\min \{k,s\}$ and plugging the identity back into \eqref{eq:sym_GDoF_region_phy}, 
we obtain  the following corollary.
\begin{corollary}
\label{corollary:symm_multicast_region}
The symmetric  $\sigma$-multicast  GDoF region  $\mathcal{D}^{\mathrm{PHY}}_{\mathrm{sym}}(\sigma,\bm{\alpha},s) $
is given by all GDoF tuples $(r_{k}: i \in [K], \ r_{\mathrm{sym}}) \in \mathbb{R}^{K+1}$ that satisfy:
\begin{equation}
\sum_{i \in [k]} r_{i} + \left[ \binom{K}{\sigma} - \binom{K - \min\{k,s\}}{\sigma} \right] \cdot r_{\mathrm{sym}}  \leq \alpha_{k}, \ 
\forall k \in [K].
\end{equation}
\end{corollary}
From  the characterization of $\mathcal{D}^{\mathrm{PHY}}_{\mathrm{sym}}(\sigma,\bm{\alpha},s) $ in the above corollary, 
it follows that for any feasible unicast GDoF tuple $\mathbf{r} = (r_{k} : k\in [K])$, we achieve any symmetric multicast GDoF that satisfies
\begin{equation}
\label{eq:max_r_sym}
r_{\mathrm{sym}} \leq \min_{k \in [K]} \left\{ \frac{ \big( \alpha_{k} - \sum_{i \in [k]} r_{i} \big) }{ \binom{K}{\sigma} - \binom{K - \min\{k,s\}}{\sigma}  }  \right\}.
\end{equation}
\section{Achievability}
\label{sec:achievability}
Equipped with the GDoF characterization for the degraded GBC with unicast and $\sigma$-multicast messages derived in the previous section, the achievability part of
Theorem \ref{theorem:optimal_trade_off} will follow from a scheme that adheres to the
separation principle, as we will see  in this section.

For content placement, generating coded multicast messages, and recovering files
from local cache contents and received multicast messages, we invoke the 
YMA scheme in  \cite{Yu2019,Yu2018}; which generalizes the original  MN scheme \cite{Maddah-Ali2014}, and reduces to it
whenever $N \geq K$.
On the other hand, the physical channel is treated as a collection of capacitated bit pipes, each carrying its corresponding coded multicast message or non-content unicast message,
at rates (or GDoF) governed by the characterization in Theorem \ref{theorem:GDoF_phy}.
We focus on  integer values of $K \mu$, drawn from $[0:K]$ in this section.
The case of non-integer $K \mu$, drawn from $(0,K)$, is treated in Appendix \ref{appendix:non_integer_K_mu}.
\subsection{Cache placement}
Each file $F_{n}$ is divided into $\binom{K}{K \mu}$ equal sized sub-files, i.e.
\begin{equation}
F_{n}  \rightarrow  \left\{ F_{n}^{\mathcal{S}'} : \mathcal{S}' \subseteq [K], |\mathcal{S}'| = K \mu \right\}
\end{equation}
where each sub-file $F_{n}^{\mathcal{S}'}$ has a size of $B / \binom{K}{K \mu}$ bits.
Each user $k$ then fills its cache memory as:
\begin{equation}
\label{eq:cache_content}
U_{k} = \left\{ F_{n}^{\mathcal{S}'} : n\in[N],  \mathcal{S}' \subseteq [K], |\mathcal{S}'| = K, k \in \mathcal{S}' \right\}.
\end{equation} 
This caching strategy satisfies the cache size constraint of $MB$ bits, (see, e.g., \cite{Maddah-Ali2014}). 
Note that the above-described procedure exactly matches the original MN uncoded caching procedure in \cite{Maddah-Ali2014}, which 
 is clearly independent of user demand tuples.
We now proceed to describe the coded multicasting and transmission procedures, which depend on the demand tuple type.
\subsection{Coded multicast messages}
\label{subsec:worst_case_demands}
Let us recall from Remark \ref{remark:worst_case} that we consider worst-case demand tuples that comprise of the maximum possible 
number of distinct user demands, i.e. $\min \{K,N \} $.
For ease of exposition, we start by focusing on the case where these distinct demands are made by the first (i.e. weakest) $\min \{K,N \} $ users. 
In Appendix \ref{appendix:other_demands}, we show that the performance achieved in this case is also achievable whenever the distinct demands
are not necessarily made by the weakest users.

We refer to $\mathcal{U} =  \big[\min \{K,N \} \big]$ as the set of \emph{leading users}, where such users request distinct files, 
while the set of \emph{non-leading users} is given by $\bar{\mathcal{U}} = [K] \setminus \mathcal{U} $.
For brevity, we use the physical channel notation from the previous section and set the multicast group size to $\sigma = K \mu + 1$, and  the number of distinct demands to  $s = \min\{K,N\}$. 
Once demands are revealed, the transmitter generates $\binom{K}{\sigma} - \binom{K -s}{\sigma}$ coded multicast messages,
each intended to a unique subset of  $\sigma$ users denoted by $\mathcal{S}$, where $\mathcal{S} \in \cup_{i \in [s]} \Sigma_{i}$.
It can be verified that each such subset of users, i.e.  $\mathcal{S} \in \cup_{i \in [s]} \Sigma_{i} $, contains at least one leading user from $\mathcal{U} = [s]$.
The coded multicast message corresponding to $\mathcal{S}$ is given by
\begin{equation}
W_{\mathcal{S}} = \bigoplus_{i \in {\mathcal{S}}}  F_{d_{i}}^{\mathcal{S} \setminus \{i\}}.
\end{equation}

Assuming the successful delivery of coded multicast messages, each leading user $k \in \mathcal{U}$ recovers the requested file $F_{d_{k}}$ from the the set of  coded  multicast messages $\{W_{\mathcal{S}} :  \mathcal{S} \in \cup_{i \in [s]}\Sigma_{i}, k \in \mathcal{S} \}$ and the cache content $U_{k}$, using the standard MN decoding procedure.
In particular, each  $W_{\mathcal{S}}$ with $k \in \mathcal{S} $ may be expressed as 
$W_{\{ k\} \cup \mathcal{S}'} = F_{d_{k}}^{\mathcal{S}' } \oplus \big(  \bigoplus_{i \in \mathcal{S}' }  F_{d_{i}}^{\{k\} \cup \mathcal{S}' \setminus \{i\}}  \big)$, where $\mathcal{S}' = \mathcal{S} \setminus \{k\}$, 
from which  undesired sub-files can be cancelled out, as they are available in $U_{k}$.
For scenarios where $N \geq K  - \sigma + 1$, coded multicast messages corresponding to all subsets of 
$\sigma$ users are transmitted, and non-leading users decode their requested files according to the above procedure.

For scenarios  where $N \leq K - \sigma$, only a subset of coded multicast messages is transmitted, i.e. those useful to leading users.
Nevertheless, non-leading users can also recover their requested files using the YMA decoding procedure \cite{Yu2018}, subject to the successful decoding of required multicast messages, discussed further on.
In particular, a non-leading user $k \in  \bar{\mathcal{U}}$ computes the   \emph{missing} set coded multicast messages, that is $\{ W_{\mathcal{A}} : \mathcal{A} \subseteq \bar{\mathcal{U}} ,  |\mathcal{A} | = \sigma , k \in \mathcal{A}  \}$, from a subset of the transmitted multicast messages 
and then proceeds to recover $F_{d_{k}}$ using the standard MN decoding procedure.
Each missing message $W_{\mathcal{A}} $ is computed by users in $\mathcal{A}$ as
\begin{equation}
\label{eq:W_A_YMA}
W_{\mathcal{A}}   = \bigoplus_{\mathcal{V}  \in \Upsilon} W_{\mathcal{B} \setminus \mathcal{V}}
\end{equation}
where $\mathcal{B} = \mathcal{A} \cup \mathcal{U}$, and  $\Upsilon$ denotes a family of subsets of  $\mathcal{B}$ such that each 
member $\mathcal{V}  \in \Upsilon$ is a set of $N$ users with distinct demands, and $\mathcal{V} \neq \mathcal{U}$;  i.e.
each $\mathcal{V}$ is a potential set of leaders other than $ \mathcal{U}$.
For more details about the YMA procedure, readers are referred to  \cite[Sec. IV.B]{Yu2018}.
\subsection{Transmission}
The problem now reduces to delivering the set of $\binom{K}{\sigma} - \binom{K - s}{\sigma}$ coded multicast messages given by $\{W_{\mathcal{S}} :  \mathcal{S} \in \cup_{i \in [s]}\Sigma_{i} \}$, as well as the set of unicast messages 
$\{W_{k} :  k \in  [K] \}$.
This is exactly the unicast and $\sigma$-multicast transmission problem discussed in Section \ref{sec:degraded_GBC_unicast_multicast}.
Moreover, in scenarios where $N \leq K  - \sigma$, the degradedness of the physical channel guarantees that each non-leading user in 
$\bar{\mathcal{U}} = [N+ 1 : K]$ can recover the entire set of multicast messages $\{W_{\mathcal{S}} :  \mathcal{S} \in \cup_{i \in [s]}\Sigma_{i} \}$. This, in turn, ensures the success of the YMA decoding procedure for such users.

Using the symmetric  $\sigma$-multicast transmission with only a subset of multicast messages in Section \ref{subsec:symm_multicast_subset},
for any achievable tuple $(\mathbf{r},r_{\mathrm{sym}}) \in \mathcal{D}^{\mathrm{PHY}}_{\mathrm{sym}}(\sigma,\bm{\alpha},s) $,
each of the non-content unicast messages achieves its corresponding GDoF in $\mathbf{r}$, while the achievable content GNDT is given by
\begin{equation}
\label{eq:achievable_tau}
\tau = \frac{1}{r_{\mathrm{sym}} \cdot \binom{K}{\sigma - 1}}.
\end{equation}
Note that the normalization factor in \eqref{eq:achievable_tau} appears since each coded multicast message $W_{\mathcal{S}} $ has a size of 
$1 / \binom{K}{\sigma - 1}$ when normalized by the file size $B$.  
Combining with \eqref{eq:max_r_sym}, we have 
\begin{equation}
\label{eq:tau_achievability}
\tau \geq \max_{k \in [K]} \left\{ 
\frac{ 1}{ \big( \alpha_{k} - \sum_{i \in [k]} r_{i} \big)  }  \cdot \frac{ \binom{K}{\sigma} - \binom{K - \min\{k,s\}}{\sigma}  }{  \binom{K}{\sigma  - 1}  }  \right\}.
\end{equation}
In \eqref{eq:tau_achievability}, we have $\min\{k,s\} = \min\big\{k, \min\{ K,N \} \big\} = \min \{k,N \}$, for all $k \in [K]$.
Therefore, \eqref{eq:tau_achievability} coincides with \eqref{eq:tau_ub}, which completes the proof of achievability in this case.
\begin{remark}
By eliminating non-content messages and restricting to the case of $N \geq K$, the achievability scheme proposed in this paper reduces to the one in \cite{Lampiris2019}. Nevertheless, the proof here is different, specifically the part dealing with transmission over the physical channel.
In \cite{Lampiris2019}, an explicit power allocation strategy is constructed to achieve the corresponding GNDT.
In this paper, we avoid the power allocation problem all together by eliminating the power allocation variables using a Fourier-Motzkin elimination procedure (see Appendix \ref{appendix:GDoF_region}). This enables us to take the additional step of  characterizing the entire GDoF region for the physical channel with  unicast and  $\sigma$-multicast (Section \ref{sec:degraded_GBC_unicast_multicast}), and leads to the GNDT-GDoF trade-off in \eqref{eq:tau_achievability}.
\end{remark}

\section{Converse}
\label{section:converse}
\label{sec:converse}
In this section, we prove the converse part of Theorem \ref{theorem:optimal_trade_off}.
We focus on worst-case demands as defined in Section \ref{subsec:worst_case_demands}.
For any such demand tuple $\mathbf{d}$, each user $k$ in $[K]$ must recover both the message $W_{k}$ and the demanded file $F_{d_k}$
from the received signal $Y_{k}^{T}$ and cache content $U_{k}$,
with a decoding error that vanishes as  $T$ grows large.
Therefore,  Fano's inequality implies:
\begin{align}
H\big( W_{k}, F_{d_k} | Y_{k}^{T} , U_{k} \big)  \leq  1 + P_{e,T} (TR_{k} + B) = T \epsilon_{T}
\end{align}
where both $P_{e,T}$ and $\epsilon_{T}$ approach zero as $T$ approaches infinity.
Let us now define a side information variable $S_{k}$ which is independent of $W_{k}$.
The side information $S_{k}$ is provided to user $k$ through a genie, and will be specified later on.
It follows that
\begin{align}
\nonumber
TR_{k} + H\big( F_{d_k} | U_{k} , S_{k}  \big)   & = H\big( W_{k} \big) + H\big( F_{d_k} | U_{k} , S_{k}  \big)   \\
\nonumber
 & = H\big( W_{k}, F_{d_k} | U_{k} , S_{k}  \big)   \\
\nonumber
 & = I \big( W_{k}, F_{d_k}  ;  Y_{k}^{T}  | U_{k} , S_{k}  \big) + H\big( W_{k}, F_{d_k} | Y_{k}^{T} , U_{k} , S_{k}  \big) \\
\label{eq:single_user_bound}
& \leq  I \big( W_{k}, F_{d_k}  ;  Y_{k}^{T}  | U_{k} , S_{k}  \big) +  T \epsilon_{T}.
\end{align}
Now let us consider a subset of users $[s]$ with distinct demands, for some $s \in [\min\{K,N\}]$.
From  the  single-user bound in \eqref{eq:single_user_bound}, we obtain a multi-user bound for such subset as
\begin{equation}
\label{eq:sum_H_I}
\sum_{k = 1}^{s}  H\big( F_{d_k} | U_{k} , S_{k}  \big)  \leq \sum_{k = 1}^{s} 
\Big[ I \big( W_{k}, F_{d_k}  ;  Y_{k}^{T}  | U_{k} , S_{k}  \big)  -  T \big( R_{k} -  \epsilon_{T} \big) \Big].
\end{equation} 
Next, we wish to find an upper bound for the right-hand-side of \eqref{eq:sum_H_I}, and a lower bound for the 
left-hand-side of the same inequality.
To this end, we apply a symmetrization step over file demands and user orders, which 
is required to bound below the left-hand-side in \eqref{eq:sum_H_I}.

Let $p:[s] \rightarrow [s]$ be a permutation over the subset of users $[s]$,
and $\mathcal{P}_{s}$ be the corresponding set of all $s!$ user permutations.
Similarly, $q:[N] \rightarrow [N]$ is  a permutation over the set of files $[N]$,
and $\mathcal{P}_{N}$ is  the corresponding set of all $N!$ file permutations.
For any pair of permutations $(p,q) \in \mathcal{P}_{s} \times \mathcal{P}_{N} $, 
suppose that each user $p(k)$ demands  the file $F_{q(k)}$.
From \eqref{eq:sum_H_I}, we write
\begin{equation}
\label{eq:sum_H_I_perm}
\sum_{k = 1}^{s}  H\big( F_{q(k)} | U_{p(k)} , S_{p(k)}  \big)  \leq  \sum_{k = 1}^{s} 
\Big[ I \big( W_{p(k)}, F_{q(k)}  ;  Y_{p(k)}^{T}  | U_{p(k)} , S_{p(k)}  \big)  -  T \big( R_{p(k)} -  \epsilon_{T} \big) \Big].
\end{equation} 
Taking the average of  both sides in \eqref{eq:sum_H_I_perm} over all possible permutations 
$(p,q) \in \mathcal{P}_{s} \times \mathcal{P}_{N} $, we obtain
\begin{multline}
\label{eq:sum_H_I_perm_av}
\frac{1}{s! N !} \sum_{(p,q) \in \mathcal{P}_{s} \times \mathcal{P}_{N}}  \sum_{k = 1}^{s}  H\big( F_{q(k)} | U_{p(k)} , S_{p(k)}  \big)  \leq \\
\frac{1}{s! N !} \sum_{(p,q) \in \mathcal{P}_{s} \times \mathcal{P}_{N}}   \sum_{k = 1}^{s} 
\Big[ I \big( W_{p(k)}, F_{q(k)}  ;  Y_{p(k)}^{T}  | U_{p(k)} , S_{p(k)}  \big)  -  T \big( R_{p(k)} -  \epsilon_{T} \big) \Big].
\end{multline} 
In what follows, we set the side information variable $S_{p(k)}$ for each user $p(k)$ 
as
\begin{equation}
\label{eq:side_information}
S_{p(k)} =  \big(W_{p(i)} ,F_{q(i)}, U_{p(i)}  : i \in [k-1]\big)
\end{equation}
consisting of intended messages, demanded files and cache contents of all users that precede user $p(k)$ in the permutation order.
Note that the independence between $S_{p(k)}$ and  $W_{p(k)}$ is preserved. 
Moreover, when focusing on a subset of users given by $[s]$, we assume that 
$p(s+1) = s+1$ and we use  $S_{s+1}$  to denote side information that contains intended messages, demanded files and cache contents of all users in $[s]$.
Next, we separately bound each side of the inequality in \eqref{eq:sum_H_I_perm_av}.
\subsection{Bounding the right-hand-side of \eqref{eq:sum_H_I_perm_av}}
To this end, we  present the following lemma.
\begin{lemma}
\label{lemma:degraded_with_caches}
For any pair of users $k$ and $j$ in $[K]$, such that $k \leq j$, we have 
\begin{equation}
\label{eq:degraded_with_caches}
I \big( W_{k}, F_{d_k}  ;  Y_{k}^{T}  | U_{k} , S_{k}  \big)  \leq  I \big( W_{k}, F_{d_k}  ;  Y_{j}^{T}  | U_{k} , S_{k}  \big).
\end{equation}
\end{lemma}
\begin{proof}
The inequality in \eqref{eq:degraded_with_caches} follows directly from the degradedness of the physical channel.
In particular, by considering the physical channel in isolation of the caches, we have the Markov chain\footnote{For rigour, there exists 
a random variable $\tilde{Y}_{j}^{T} \sim Y_{j}^{T}$ such that \eqref{eq:Markov_degraded} holds while replacing $Y_{j}^{T}$
with $\tilde{Y}_{j}^{T}$ \cite{ElGamal2011}. Without loss of generality, we use $Y_{j}^{T}$ instead of $\tilde{Y}_{j}^{T}$ and  assume that \eqref{eq:Markov_degraded} holds.}  $\big( W_{k}, F_{d_k} \big) \rightarrow X^{T}   \rightarrow Y_{j}^{T}    \rightarrow Y_{k}^{T}$.
By providing $\big(U_{k} , S_{k}   \big) $ as side information to both users $k$ and $j$, this degradedness is not altered,
and the following Markov chain holds
\begin{equation}
\label{eq:Markov_degraded}
\big( W_{k}, F_{d_k} \big) \rightarrow \big( X^{T},U_{k} , S_{k} \big)   \rightarrow \big( Y_{j}^{T},U_{k} , S_{k}  \big)  
 \rightarrow \big( Y_{k}^{T},U_{k} , S_{k} \big).  
\end{equation}
It follows that
\begin{align}
\nonumber
I \big( W_{k}, F_{d_k}  ;  Y_{k}^{T}  | U_{k} , S_{k}  \big)  & =  I \big( W_{k}, F_{d_k}  ;  Y_{k}^{T}, U_{k} , S_{k}   \big)  -  
I \big( W_{k}, F_{d_k}  ;  U_{k} , S_{k}   \big)  \\
\label{eq:degradedness_proof}
&  \leq  I \big( W_{k}, F_{d_k}  ;  Y_{j}^{T}, U_{k} , S_{k}   \big)  -  
I \big( W_{k}, F_{d_k}  ;  U_{k} , S_{k}   \big)  \\
&  =  I \big( W_{k}, F_{d_k}  ;  Y_{j}^{T} | U_{k} , S_{k}   \big)
\end{align}
where the inequality in \eqref{eq:degradedness_proof} is due to \eqref{eq:Markov_degraded} and the data processing inequality.
\end{proof}
Equipped with the above lemma and focusing on an arbitrary permutation pair  $(p,q) \in \mathcal{P}_{s} \times \mathcal{P}_{N} $, the corresponding term on the right-hand-side of 
\eqref{eq:sum_H_I_perm_av} is bounded as:
\begin{align}
\nonumber
 \sum_{k = 1}^{s}  I \big( W_{p(k)}, F_{q(k)}  ;  Y_{p(k)}^{T}  | U_{p(k)} , S_{p(k)}  \big) & \leq 
  \sum_{k = 1}^{s}  I \big( W_{p(k)}, F_{q(k)}  ;  Y_{s}^{T}  | U_{p(k)} , S_{p(k)}  \big)  \\
  \nonumber
  & =  \!   \sum_{k = 1}^{s}  \!  h \big( Y_{s}^{T}  | U_{p(k)} , S_{p(k)}  \big)  \! -  \!
   h \big(   Y_{s}^{T}  | U_{p(k)} , S_{p(k)}, W_{p(k)}, F_{q(k)}  \big)  \\
   \nonumber
 & \leq   \sum_{k = 1}^{s-1}  h \big( Y_{s}^{T}  | U_{p(k)} , S_{p(k)}  \big)  -
   h \big(   Y_{s}^{T}  | U_{p(k+1)},S_{p(k+1)} \big)   \\ 
    \label{eq:MI_ub_converse_0_1}
 & \quad + h \big( Y_{s}^{T}  | U_{p(s)} , S_{p(s)}  \big)  -
   h \big(   Y_{s}^{T}  | S_{s+1} \big) \\
    \label{eq:MI_ub_converse_0_2}
 & =    h \big( Y_{s}^{T}  | U_{p(1)}  \big)  -
   h \big(   Y_{s}^{T}  | S_{s+1} \big)   \\
 \label{eq:MI_ub_converse_0}
 & =    h \big( Y_{s}^{T}  | U_{p(1)}  \big)  -
   h \big(   Z_{s}^{T}  \big)   \\
    \nonumber
 & =    I \big( X^{T} ; Y_{s}^{T}  | U_{p(1)}  \big)    \\
 \label{eq:MI_ub_converse}
  & \leq T \log(1 + P^{\alpha_{s}}).
\end{align}
In \eqref{eq:MI_ub_converse_0_1}, we have used $S_{p(k+1)} = \big( S_{p(k)},W_{p(k)} ,F_{q(k)} ,U_{p(k)} 
\big)$ for all $k\in [s-1]$, which holds by definition; the fact that conditioning does not increase differential entropy; and the definition of  $S_{s+1} $.
The equality in \eqref{eq:MI_ub_converse_0} follows from \eqref{eq:MI_ub_converse_0_2} by recalling that $S_{s+1} $ contains messages and files intended to all users in $[s]$, mapping directly to $X^{T}$, which 
in turn is removed from $ Y_{s}^{T}  $.

As  \eqref{eq:MI_ub_converse}  holds for all permutations $(p,q) \in \mathcal{P}_{s} \times \mathcal{P}_{N}$, 
and since for any such permutation in \eqref{eq:sum_H_I_perm_av} we have  $\sum_{k \in [s]} R_{p (k)} = \sum_{k \in [s]} R_{k}$,
it follows that each of the inner sums (over $k$) on the right-hand-side of \eqref{eq:sum_H_I_perm_av}  is bounded by the same term.  
Therefore, we obtain the bound  
\begin{multline}
\label{eq:conv_caches_mutual_inf}
\frac{1}{s! N !} \sum_{(p,q) \in \mathcal{P}_{s} \times \mathcal{P}_{N}}   \sum_{k = 1}^{s} 
\Big[ I \big( W_{p(k)}, F_{q(k)}  ;  Y_{p(k)}^{T}  | U_{p(k)} , S_{p(k)}  \big)  -  T \big( R_{p(k)} -  \epsilon_{T} \big) \Big] 
\leq \\ 
T \left[ \log(1 + P^{\alpha_{s}})  - \sum_{k = 1}^{s}  \big( R_{k} -  \epsilon_{T} \big) \right].
\end{multline}
\subsection{Bounding the left-hand-side of \eqref{eq:sum_H_I_perm_av}}
It is evident that for every $k \in [s]$ and $(p,q) \in \mathcal{P}_{s} \times \mathcal{P}_{N}$, we have  
\begin{equation}
\label{eq:conv_caching_entropy_cond}
H\big( F_{q(k)} | U_{p(k)} , S_{p(k)}  \big)  = H\big( F_{q(k)} | U_{p(1)}, \ldots,  U_{ p(k)} , F_{q(1)}, \ldots, F_{q(k-1)} \big),
\end{equation}
which holds since messages are independent of files and cache contents (see \eqref{eq:side_information}).
From the equality in \eqref{eq:conv_caching_entropy_cond}, it  can be seen that the 
 left-hand-side of \eqref{eq:sum_H_I_perm_av} is in fact a lower bound on the number of bits that must be
delivered (i.e.  load) in a conventional share-link setting with $s$ users, up to a decoding error term  \cite[eq. (30)]{Yu2019}.
We hence employ the results and techniques of \cite{Yu2019} to obtain:
\begin{align}
\label{eq:av_cond_entropy_bound_1}
\frac{1}{s! N !} \sum_{(p,q) \in \mathcal{P}_{s} \times \mathcal{P}_{N}}  \sum_{k = 1}^{s}  H\big( F_{q(k)} | U_{p(k)} , S_{p(k)}  \big)  
& \geq  B  \cdot  \left( s' - 1 + a - \frac{s' (s'-1) - l(l-1) + 2as' }{2(N - l + 1)} M  \right)  \\
\label{eq:av_cond_entropy_bound_2}
&  \geq 
 \frac{B}{2.01} \cdot \left(   \frac{N - M}{M} \big( 1 - (1  -  M / N)^{s} \big) \right) \\
 \label{eq:av_cond_entropy_bound_3}
&  \geq \frac{B}{2.01} \cdot \mathrm{conv} \left(  \frac{\binom{K}{K \mu + 1} - \binom{K-s}{K \mu + 1}}{\binom{K}{K \mu }} \right)
\end{align}
where the bound in \eqref{eq:av_cond_entropy_bound_1} holds for any parameters $s' \in [s]$ and $a \in [0,1]$, while 
$l \in [s']$ is the minimum value that satisfies: $\big(s' (s'-1) - l(l-1) + 2as' \big)/2 \leq (N - l + 1)l$.
The bound in \eqref{eq:av_cond_entropy_bound_1} follows directly from \cite[Lem. 3]{Yu2019}.
On the other hand, going from \eqref{eq:av_cond_entropy_bound_1}  to  within a multiplicative factor of $2.01$ from the decentralized load in \eqref{eq:av_cond_entropy_bound_2} holds due to \cite[Lem. 1]{Yu2019}. 
Finally, the inequality in \eqref{eq:av_cond_entropy_bound_3} follows from the results in \cite{Yu2018} (see also \cite[Appendix G]{Yu2019} where a similar step is used).\footnote{Note that the lower convex envelope in \eqref{eq:av_cond_entropy_bound_3} is defined in a similar manner to \eqref{eq:tau_ub} in Definition \ref{def:GNDT_ub}.} 
\subsection{Combining bounds}
From \eqref{eq:sum_H_I_perm_av}, \eqref{eq:conv_caches_mutual_inf} and \eqref{eq:av_cond_entropy_bound_3}, and by taking the limit 
$T \to \infty$, we obtain
\begin{equation}
\label{eq:conv_sum_rate_UB}
 \sum_{k = 1}^{s} R_{k} +  \frac{1}{2.01 \cdot\mathcal{T} } \cdot \mathrm{conv} \left(  \frac{\binom{K}{K \mu +1} - \binom{K-s}{K \mu +1}}{\binom{K}{K \mu}} \right)  \leq \log(1 + P^{ \alpha_{s} } ).
\end{equation}
In the GDoF-GNDT limit, the bound in \eqref{eq:conv_sum_rate_UB} translates to 
\begin{equation}
\label{eq:conv_sum_GDoF_UB}
 \sum_{k = 1}^{s} r_{k}  +  \frac{1}{2.01 \cdot \tau} \cdot \mathrm{conv} \left(  \frac{\binom{K}{K \mu +1} - \binom{K-s}{K \mu +1}}{\binom{K}{K \mu}} \right)  \leq  \alpha_{s}.
\end{equation}
The above holds for any $s \in [\min \{ K , N \}]$. These bounds fully describe the lower bound in \eqref{eq:order_optimality}
whenever $N \geq K$.
For $N < K$, we require the additional bounds derived next.
\subsection{Remaining bounds for $N < K$}
Let us now consider a subset of users $[s]$, for some $s \in [\min\{K,N\} + 1 : K ]$. 
Applying the exact above steps to the first $N$ users in $[s]$, which request distinct files,  
we obtain
\begin{align}
\label{eq:sum_H_I_N}
\frac{B}{2.01} \cdot \mathrm{conv} \left(  \frac{\binom{K}{K \mu + 1} - \binom{K-N}{K \mu + 1}}{\binom{K}{K \mu }} \right)+  
\sum_{k = 1}^{N} T(R_{k} - \epsilon_{T})  &  \leq 
 \sum_{k = 1}^{N}  I \big( W_{p(k)}, F_{q(k)}  ;  Y_{s}^{T}  | U_{p(k)} , S_{p(k)}  \big) \\
\label{eq:sum_H_I_N_2}
 & \leq h \big( Y_{s}^{T}  | U_{p(1)}  \big)  -
   h \big(   Y_{s}^{T}  | U_{N+1},S_{N+1} \big).
\end{align}
The bound in \eqref{eq:sum_H_I_N} follows from  \eqref{eq:sum_H_I_perm_av} after rearranging,
bounding the left-hand-side using \eqref{eq:av_cond_entropy_bound_3}, and bounding the right-hand-side 
by fixing a permutation pair $(p,q)$ that maximizes the average.
The bound in \eqref{eq:sum_H_I_N_2} follows by employing the same steps used to obtain \eqref{eq:MI_ub_converse}.

For the remaining users in $[N+1 : s]$, let us define their side information variables as
\begin{equation}
S_{k} = \big(W_{i} ,F_{d_{i}}, U_{i} : i \in [k-1] \big).
\end{equation}
We also use $S_{s+1}$ to denote a side information variable comprising of messages, requested files and cache contents for all users in $[s]$.
The non-content sum-rate is bounded above as
\begin{align}
 \label{eq:non_content_sum_rate_bound_1}
\sum_{k = N+1}^{s} T(R_{k} - \epsilon_{T})
& \leq  \sum_{k = N+1}^{s}  I \big( W_{k}, F_{d_{k}}  ;  Y_{s}^{T}  | U_{k} , S_{k}  \big)  \\
\nonumber
&\leq   \sum_{k = N+1}^{s}  h \big( Y_{s}^{T}  | U_{k} , S_{k}  \big)  -
 h \big(   Y_{s}^{T}  | U_{k+1},S_{k+1}\big)   \\ 
 \label{eq:non_content_sum_rate_bound_3}
& =  h \big( Y_{s}^{T}  | U_{N+1} , S_{N+1}  \big)  -
 h \big(   Y_{s}^{T}  | S_{s+1}\big)
\end{align}
where the inequality in \eqref{eq:non_content_sum_rate_bound_1} follows from the single user bounds in \eqref{eq:single_user_bound} and Lemma \ref{lemma:degraded_with_caches}.
By adding the bounds in \eqref{eq:sum_H_I_N_2} and \eqref{eq:non_content_sum_rate_bound_3}, we obtain
\begin{align}
\frac{B}{2.01} \cdot \mathrm{conv} \left(  \frac{\binom{K}{K \mu + 1} - \binom{K-N}{K \mu + 1}}{\binom{K}{K \mu }} \right)+  
\sum_{k = 1}^{s} T(R_{k} - \epsilon_{T})  &  \leq 
h \big( Y_{s}^{T}  | U_{p(1)}  \big)  - h \big(   Y_{s}^{T}  | S_{s+1} \big)  \\
  &  \leq  T \log(1 + P^{\alpha_{s}})
\end{align}
which in the GDoF-GNDT limit, translates to 
\begin{equation}
\label{eq:conv_sum_GDoF_UB_N}
 \sum_{k = 1}^{s} r_{k}  +  \frac{1}{2.01 \cdot \tau} \cdot \mathrm{conv} \left(  \frac{\binom{K}{K \mu +1} - \binom{K-N}{K \mu +1}}{\binom{K}{K \mu}} \right)  \leq  \alpha_{s}.
\end{equation}
The bound in \eqref{eq:conv_sum_GDoF_UB_N}  holds for all $s \in [N+1 : K]$. 
By rearranging the terms in \eqref{eq:conv_sum_GDoF_UB} and \eqref{eq:conv_sum_GDoF_UB_N}, and taking the tightest of such bounds over all $s \in [K]$, we obtain a lower bound given by
\begin{align}
\label{eq:conv_tau_LB}
2.01 \cdot \tau  & \geq \max_{s \in [K]} \left\{    \frac{1}{\big(\alpha_{s} - \sum_{k = 1}^{s} r_{s} \big)} \cdot \mathrm{conv} \left(  \frac{\binom{K}{K \mu +1} - \binom{K-\min\{ s, N \}}{K \mu +1}}{\binom{K}{K \mu}} \right) \right\}
\end{align}
where  $ \sum_{k = 1}^{s} r_{k}  \leq \alpha_{s}$ for all $s \in [K]$.
The right-hand-side of \eqref{eq:conv_tau_LB} coincides with  $\tau^{\mathrm{ub}}(\mathbf{r} ; \mu,  \bm{\alpha})$ in \eqref{eq:tau_ub}.
This completes the converse proof.
\section{Conclusion}
In this work, we introduced the problem of wireless coded caching under mixed cacheable content and uncacheable non-content types of traffic.
Focusing on networks in which the physical channel is modelled by a degraded GBC, 
we  proposed a caching and delivery strategy based on the separation principle, which isolates the coded caching and multicasting problem from the physical layer transmission problem.
We proved that the proposed strategy achieves near optimal performances in the information-theoretic sense.
Through our analysis, we revealed \emph{topological holes} arising due to asymmetries in wireless network topologies, 
which enable the transmission of non-content messages while incurring no loss in terms of content delivery time. 
The extension of this result to other networks, including multi-transmitter and multi-antenna networks, is of high interest. 
In such networks, the performance is characterized not only by channel strength parameters (i.e. topology), but also by
the quality of channel state information at the transmitters (CSIT)---see, e.g., \cite{Zhang2015,Zhang2017,Piovano2017,Ngo2018,Lampiris2018a,
Piovano2019,Lampiris2017,Piovano2020,Bergel2018,Sengupta2017,Zhang2019}.
This leads to an explosion in the number of system parameters in general (i.e. channel strengths and CSIT qualities), rendering the corresponding problems extremely challenging.
One way to control the number of system parameters is to enforce symmetry (e.g. equal channel strengths, CSIT qualities, or both), as done in most of the aforementioned works. 
Nevertheless, apart from being an oversimplification, symmetry also obscures the role of topological holes, whose study necessarily requires venturing beyond symmetric settings.
The prospect of unveiling the role of topological holes in asymmetric multi-transmitter and multi-antenna  cache-aided networks with mixed traffic is both intriguing and not 
yet explored.

\appendix
\section*{Appendices}
\section{Unicast and Multiple Multicast GDoF Region}
\label{appendix:GDoF_region}
In this appendix, we present a proof for Theorem \ref{theorem:GDoF_phy}.
\subsection{Converse}
Starting with the converse, we invoke Fano's inequality from which we obtain: 
\begin{align}
\label{eq:Fano_GBC_1}
T(R_{i} - \epsilon_{T}) + T \sum_{\mathcal{S} \in \Sigma_{i}} (R_{\mathcal{S}} -
\epsilon_{T})  & \leq   I
 \big( W_{i}, \{ W_{\mathcal{S}} : \mathcal{S} \in \Sigma_{i} \} ; Y_{i}^{T} \big) \\
 \label{eq:Fano_GBC_2}
 & \leq  I
 \big( W_{i}, \{ W_{\mathcal{S}} : \mathcal{S} \in \Sigma_{i} \} ; Y_{k}^{T}  \big)  \\
 \label{eq:Fano_GBC_3}
 & \leq  I
 \big( W_{i}, \{ W_{\mathcal{S}} : \mathcal{S} \in \Sigma_{i} \} ; Y_{k}^{T} | W_{i}'   \big)
\end{align}
where $T$ is the number of channel uses over which the communication occurs,
$\epsilon_{T}$ is an error term that approaches zero as $T \to \infty$, 
and $W_{i}'  \triangleq  \{W_{j}, W_{\mathcal{S}} : 
\mathcal{S} \in \Sigma_{j} , j \in [i - 1]  \}$ is a side information variable.
The inequality in \eqref{eq:Fano_GBC_2} holds for all $k \geq i$ due to the degradedness of the physical channel and the order in \eqref{eq:order_SNR}, while \eqref{eq:Fano_GBC_3} holds since $W_{i}'$ is independent of $W_{i}$
and $\{ W_{\mathcal{S}} : \mathcal{S} \in \Sigma_{i} \}$.

For any $k  \in [K]$,  adding up the bounds obtained from \eqref{eq:Fano_GBC_3} for all $i \in [k]$, we obtain 
\begin{align}
\nonumber
T \sum_{i \in [k]} (R_{i} - \epsilon_{T})  + T \sum_{\mathcal{S} \in \cup_{i \in [k]} \Sigma_{i}} (R_{\mathcal{S}} - \epsilon_{T})   & \leq 
\sum_{i \in [k]} I \big( W_{i}, \{ W_{\mathcal{S}} : \mathcal{S} \in \Sigma_{i} \} ; Y_{k}^{T} | W_{i}'   \big)
\\
\nonumber
& =   I
 \big( \{ W_{i}, W_{\mathcal{S}} : \mathcal{S} \in \Sigma_{i}, i \in [k] \} ; Y_{k}^{T} \big)  \\ 
  \label{eq:Fano_GBC_6}
 & \leq T \log (1 + P^{\alpha_{k}})
\end{align}
where we implicitly assume that $\Sigma_{i} = \emptyset$ for all 
$i \in [K - \sigma + 2 : K]$.
From the bound in \eqref{eq:Fano_GBC_6},
it follows that the capacity region  $\mathcal{C}^{\mathrm{PHY}}(\sigma,\bm{\alpha},P)$ is contained in the outer bound
$\mathcal{C}^{\mathrm{PHY}}_{\mathrm{out}}(\sigma,\bm{\alpha},P)$, described by all rate tuples 
$(R_{k}: k \in [K], R_{\mathcal{S}} : \mathcal{S} \in \Sigma ) \in \mathbb{R}_{+}^{K + \binom{K}{\sigma}}$ satisfying:
\begin{equation}
\label{eq:capacity_outer_phy}
\begin{aligned}
\sum_{i \in [k]} R_{i} + \sum_{\mathcal{S} \in \cup_{i \in [k]} \Sigma_{i} }  R_{\mathcal{S}}  & \leq \log(1 + P^{\alpha_{k}}), \ \forall k \in [K - \sigma + 1]\\
\sum_{i \in [k]}  R_{i} + \sum_{\mathcal{S} \in \Sigma}  R_{\mathcal{S}}  & \leq \log(1 + P^{\alpha_{k}}), \ \forall k \in [K - \sigma + 2 : K].
\end{aligned}
\end{equation}
In the GDoF sense, $\mathcal{C}^{\mathrm{PHY}}_{\mathrm{out}}(\sigma,\bm{\alpha},P)$ translates to the outer bound denoted by 
$\mathcal{D}^{\mathrm{PHY}}_{\mathrm{out}}(\sigma,\bm{\alpha})$, which coincides with the region
characterized by the inequalities in \eqref{eq:GDoF_region_phy}.
\subsection{Achievability}
For the achievability, we use message combining and superposition coding at the transmitter, and successive decoding at the receivers.
In particular, we construct $K$ codewords as:
\begin{equation}
\begin{aligned}
\big\{ W_{k}, W_{\mathcal{S}} : \mathcal{S} \in \Sigma_k \big\} & \rightarrow X_{k}^{T},   \ \forall k \in [K - \sigma + 1] \\
W_{k} & \rightarrow X_{k}^{T},   \ \forall k \in [K - \sigma + 2 : K] 
\end{aligned}
\end{equation}
where each combined message $\{W_{k}, W_{\mathcal{S}} : \mathcal{S} \in \Sigma_k \} $ has a rate of  
${R_{k} + \sum_{ \mathcal{S} \in \Sigma_k } R_{\mathcal{S}}}$, and 
each codeword $X_{k}^{T} \triangleq \big(X_{k}(1) , \ldots, X_{k}(T)\big)$ is drawn from an independent Gaussian codebook with unit average power.
The transmit signal $X^{T} \triangleq \big(X(1) , \ldots, X(T)\big)$ is then constructed as:
\begin{equation}
X(t) = \sum_{ k\in [K]} \sqrt{q_{k}} X_{k} (t) 
\end{equation}
where $q_{k} \geq 0$ is the power allocated to the $k$-th codeword, such that $\sum_{k \in [K]} q_{k}  \leq  1$.
On the other end, each user $k$ receives the noisy signal: $Y_{k}(t) =  \sqrt{P^{\alpha_k}} \sum_{ i \in [K]} \sqrt{q_{i}} X_{i} (t) + Z_{k}(t)$.

Each user $k \in [K]$ decodes the signals $X_{1}^{T}, X_{2}^{T}, \ldots, X_{k}^{T}$, successively in that order. 
Assuming successful decoding, each user $k$  recovers all messages in
\begin{equation}
\label{eq:decoded_message_set}
\big\{ W_{i}: i \in [k], \ W_{\mathcal{S}} : \mathcal{S}\in \cup_{i \in [ k ]} \Sigma_{i} \big\}
\end{equation}
which includes  all messages desired by user $k$,
i.e. $\{ W_{k}, W_{\mathcal{S}} : \mathcal{S}\in \Sigma, k \in \mathcal{S} \}$.
From the above, it can be seen that each codeword $X_{k}^{T}$ is decoded by all users in $[k:K]$, while treating
interference from $X_{k+1}^{T},\ldots,X_{K}^{T}$ as noise.
Therefore, messages encoded in the signal $X_{k}^{T}$ achieve all rates with a sum not exceeding
\begin{equation}
\label{eq:sum_rate_user_k}
\min_{i \in [k:K]} \left\{ \log \left(1 +  \frac{P^{\alpha_{i}} q_{i} }{1 + P^{\alpha_{i}}\sum_{j \in [k+1:K]} q_{j} } \right)  \right\} = 
\log \left(1 +  \frac{P^{\alpha_{k}} q_{k} }{1 + P^{\alpha_{k}}\sum_{j \in [k+1:K]} q_{j} } \right)
\end{equation}
where the above equality follows from the fact that $\alpha_{k} \leq \alpha_{i}$, for all $i \in [k:K]$.
The above described strategy hence achieves the rate region described by all non-negative rate tuples that
satisfy
\begin{equation}
\begin{aligned}
R_{k} + \sum_{\mathcal{S} \in \Sigma_{k} } R_{\mathcal{S}}  & \leq 
\log \left(1 +  \frac{P^{\alpha_{k}} q_{k} }{1 + P^{\alpha_{k}}\sum_{j \in [k+1:K]} q_{j} } \right), \ \forall k \in [K - \sigma + 1]\\
R_{k}  & \leq \log \left(1 +  \frac{P^{\alpha_{k}} q_{k} }{1 + P^{\alpha_{k}}\sum_{j \in [k+1:K]} q_{j} } \right) , \ \forall k \in [K - \sigma + 2 : K]
\end{aligned}
\end{equation}
for some feasible power allocation $\mathbf{q} \triangleq (q_{1},\ldots,q_{K})$.

We now obtain an inner bound on the above achievable rate region
which is more malleable for GDoF and constant-gap analysis.
To this end, we adopt the following power allocation:
\begin{equation}
\label{eq:power_allocation}
\begin{aligned}
q_{k}  & = P^{-\beta_{k}} - P^{-\beta_{k+1}}, \forall k \in [K-1]  \\
q_{K} &  = P^{-\beta_{K}}
\end{aligned}
\end{equation}
where the sequence of power exponents in \eqref{eq:power_allocation} satisfies:
\begin{align}
\label{eq:beta_variables}
0 = \beta_{1} \leq \beta_{2} \leq \cdots \leq \beta_{K} \ \ \text{and} \ \ \beta_{k+1} \leq \alpha_{k}, \ \forall k\in[K-1].
\end{align}
Recalling that $P > 1$, it can be verified that the above power allocation is feasible, and satisfies:
\begin{align}
\sum_{j \in [k:K]} q_{j}  = P^{-\beta_{k}}, \ \forall k\in[K].
\end{align}
Using this power allocation, the rate in \eqref{eq:sum_rate_user_k} is bounded below for all $k \in [K -1]$ as:
\begin{align}
\log \left(1 +  \frac{P^{\alpha_{k}} q_{k} }{1 + P^{\alpha_{k}}\sum_{j \in [k+1:K]} q_{j} } \right) & = 
\log \left(\frac{1 + P^{\alpha_{k}}\sum_{i \in [k:K]} q_{j} } {1 + P^{\alpha_{k}}\sum_{j \in [k+1:K]} q_{j} } \right)   \\
& = \log \left(\frac{1 + P^{\alpha_{k}} P^{-\beta_{k}}  } {1 + P^{\alpha_{k}} P^{-\beta_{k+1}}} \right)   \\
& \geq\ \log \left( \max \left\{ 1,  \frac{P^{\beta_{k+1}-\beta_{k}}  } {2} \right\} \right)   \\
& = \big(  (\beta_{k+1}-\beta_{k}) \log (P )  - 1  \big)^{+}.
\end{align}
For $k = K$, we obtain the same bound by setting $\beta_{K+1} = \alpha_{K}$, i.e.
\begin{equation}
\log \left(1 + P^{\alpha_{K} } q_{K}  \right)  = \log \left(1 + P^{\beta_{K+1} - \beta_{K} }  \right) \geq  
\big( (\beta_{K+1}  - \beta_{K} )\log\left( P  \right) - 1 \big)^{+}.
\end{equation}
This yields the inner bound  $\mathcal{C}^{\mathrm{PHY}}_{\mathrm{in}}(\sigma,\bm{\alpha},P)$, 
described by all non-negative rate tuples that satisfy: 
\begin{equation}
\label{eq:capacity_inner_phy_beta}
\begin{aligned}
R_{k} + \sum_{\mathcal{S} \in \Sigma_{k} } R_{\mathcal{S}}  & \leq 
\big( (\beta_{k+1}-\beta_{k}) \log \left(P  \right)  - 1 \big)^{+}, \ \forall k \in [K - \sigma + 1]\\
R_{k}  & \leq \big( (\beta_{k+1}-\beta_{k}) \log \left(P  \right)  - 1 \big)^{+}, \ \forall k \in [K - \sigma + 2 : K]
\end{aligned}
\end{equation}
for some feasible power exponents $\bm{\beta} \triangleq (\beta_{1},\ldots,\beta_{K})$, as defined in \eqref{eq:beta_variables}. 
In the GDoF sense, $\mathcal{C}^{\mathrm{PHY}}_{\mathrm{in}}(\sigma,\bm{\alpha},P)$ translates to 
$\mathcal{D}^{\mathrm{PHY}}_{\mathrm{in}}(\sigma,\bm{\alpha})$, 
described by all non-negative GDoF tuples that satisfy:
\begin{equation}
\label{eq:GDoF_inner_beta}
\begin{aligned}
r_{k} + \sum_{\mathcal{S} \in \Sigma_{k} } r_{\mathcal{S}}  & \leq 
\beta_{k+1}-\beta_{k}, \ \forall k \in [K - \sigma + 1]\\
r_{k}  & \leq \beta_{k+1}-\beta_{k}, \ \forall k \in [K - \sigma + 2 : K]
\end{aligned}
\end{equation}
for some feasible power allocation  $\bm{\beta}$.
By definition, we have $\mathcal{D}^{\mathrm{PHY}}_{\mathrm{in}}(\sigma,\bm{\alpha}) \subseteq \mathcal{D}^{\mathrm{PHY}}_{\mathrm{out}}(\sigma,\bm{\alpha})$. 
Nevertheless, it turns out that the two regions coincide as shown through the following result.
\begin{lemma}
\label{lemma:FM_elimination}
The achievable GDoF region $\mathcal{D}^{\mathrm{PHY}}_{\mathrm{in}}(\sigma,\bm{\alpha})$ and the outer bound $\mathcal{D}^{\mathrm{PHY}}_{\mathrm{out}}(\sigma,\bm{\alpha})$ are equal. 
\end{lemma}
Lemma \ref{lemma:FM_elimination} is proved by eliminating all power allocation 
variables in \eqref{eq:GDoF_inner_beta}  using means of Fourier-Motzkin elimination.
This yields an equivalent representation of $\mathcal{D}^{\mathrm{PHY}}_{\mathrm{in}}(\sigma,\bm{\alpha})$  that coincides with the inequalities in \eqref{eq:GDoF_region_phy}, and hence  $\mathcal{D}^{\mathrm{PHY}}_{\mathrm{out}}(\sigma,\bm{\alpha})$.
It follows that
\begin{equation}
\mathcal{D}^{\mathrm{PHY}}_{\mathrm{in}}(\sigma,\bm{\alpha}) = \mathcal{D}^{\mathrm{PHY}}_{\mathrm{out}}(\sigma,\bm{\alpha}) = \mathcal{D}^{\mathrm{PHY}}(\sigma,\bm{\alpha})
\end{equation}
which completes the proof of Theorem \ref{theorem:GDoF_phy}.
Next, we present the proof of Lemma \ref{lemma:FM_elimination}. 
\subsection{Proof of Lemma \ref{lemma:FM_elimination}} 
For convenience, let us define the new GDoF variables $\bm{\rho} \triangleq (\rho_{1}, \ldots, \rho_{K})$ as 
\begin{equation}
\label{eq:rho_def}
\begin{aligned}
\rho_{k} & \triangleq r_{k} + \sum_{\mathcal{S} \in \Sigma_{k} } r_{\mathcal{S}}, \ \forall k \in [1:K - \sigma + 1] \\
\rho_{k} & \triangleq r_{k}, \ \forall k \in [K - \sigma + 2:K].
\end{aligned}
\end{equation}
The achievable GDoF region $\mathcal{D}^{\mathrm{PHY}}_{\mathrm{in}}(\sigma,\bm{\alpha})$ is described by the following sets of inequalities:
\begin{align}
\label{eq:GDoF_inner_rho_beta_1}
\rho_{k} & \leq  \beta_{k+1} - \beta_{k}, \ \forall k \in [K] \\
\label{eq:GDoF_inner_rho_beta_2}
0 & \leq \alpha_{k} - \beta_{k+1}, \ \forall k \in [K-1] \\
\label{eq:GDoF_inner_rho_beta_3}
0 & \leq \beta_{k+1} - \beta_{k}, \ \forall k \in [K-1]
\end{align} 
which capture both GDoF conditions in \eqref{eq:GDoF_inner_beta},  as well as conditions on power allocation variables in
\eqref{eq:beta_variables}. 
Note that since $\rho_{k} \geq 0$, the set of inequalities in \eqref{eq:GDoF_inner_rho_beta_3} is redundant and hence can be ignored. 
We now proceed to eliminate $\beta_{2}, \ldots , \beta_{K}$ (recall that $\beta_{1} = 0$ and $\beta_{K+1} = \alpha_{K}$) using a Fourier-Motzkin procedure, see, e.g., 
 \cite[Appendix D]{ElGamal2011}.
This is carried out sequentially, eliminating 
$\beta_{K},\beta_{K-1}, \ldots , \beta_{2}$  in that order.
Starting with $\beta_{K}$, relevant inequalities are given by
\begin{align}
\label{eq:FM_elimination_1}
0 & \leq \alpha_{K-1} - \beta_{K} \\
\label{eq:FM_elimination_2}
\rho_{K} & \leq \alpha_{K} - \beta_{K} \\
\label{eq:FM_elimination_3}
\rho_{K-1} & \leq \beta_{K} - \beta_{K-1}.
\end{align}
We eliminate $\beta_{K} $ by adding each of the inequalities with $-\beta_{K} $ on the right-hand-side, i.e. \eqref{eq:FM_elimination_1} and \eqref{eq:FM_elimination_2}, to the inequality with $\beta_{K} $ on the right-hand-side, i.e. \eqref{eq:FM_elimination_3}. 
This yields
\begin{equation}
\begin{aligned}
\rho_{K-1} & \leq \alpha_{K-1} - \beta_{K-1} \\
\rho_{K} + \rho_{K-1} & \leq \alpha_{K} - \beta_{K-1}.
\end{aligned}
\end{equation}
After the elimination of $\beta_{K}$, we are left with the following inequalities
\begin{equation}
\begin{aligned}
\rho_{k} & \leq  \beta_{k+1} - \beta_{k}, \ \forall k \in [K-2] \\
\rho_{K-1} & \leq \alpha_{K-1} - \beta_{K-1} \\
\rho_{K} + \rho_{K-1} & \leq \alpha_{K} - \beta_{K-1} \\
0 & \leq \alpha_{k} - \beta_{k+1}, \ \forall k \in [K-2].
\end{aligned} 
\end{equation}
Next, we eliminate $\beta_{K-1}$. To this end, we isolate the following inequalities 
\begin{equation}
\begin{aligned}
0 & \leq \alpha_{K-2} - \beta_{K-1} \\
\rho_{K-1} & \leq \alpha_{K-1} - \beta_{K-1} \\
\rho_{K} + \rho_{K-1} & \leq \alpha_{K} - \beta_{K-1} \\
\rho_{K-2} & \leq \beta_{K-1} - \beta_{K-2}
\end{aligned}
\end{equation}
from which we eliminate $\beta_{K-1}$ and obtain
\begin{equation}
\begin{aligned}
\rho_{K-2} & \leq \alpha_{K-2} - \beta_{K-2} \\
\rho_{K-1} + \rho_{K-2} & \leq \alpha_{K-1} - \beta_{K-2} \\
\rho_{K} + \rho_{K-1} + \rho_{K-2}  & \leq \alpha_{K} - \beta_{K-2}.
\end{aligned}
\end{equation}
After eliminating $\beta_{K-1}$, we are left with the following set of inequalities
\begin{equation}
\begin{aligned}
\rho_{k} & \leq  \beta_{k+1} - \beta_{k}, \ \forall k \in [K-3] \\
\rho_{K-2} & \leq \alpha_{K-2} - \beta_{K-2} \\
\rho_{K-1} + \rho_{K-2} & \leq \alpha_{K-1} - \beta_{K-2} \\
\rho_{K} + \rho_{K-1} + \rho_{K-2}  & \leq \alpha_{K} - \beta_{K-2} \\
0 & \leq \alpha_{k} - \beta_{k+1}, \ \forall k \in [K-3].
\end{aligned} 
\end{equation}
Proceeding in a similar manner, it can be verified that after the $E$-th elimination,
where $E \in [K - 2]$, we are left with the following set of inequalities: 
\begin{equation}
\begin{aligned}
\rho_{k} & \leq  \beta_{k+1} - \beta_{k}, \ \forall k \in [K-E-1] \\
\rho_{K-E} & \leq \alpha_{K-E} - \beta_{K-E} \\
\vdots \\
\rho_{K} + \rho_{K-1} + \cdots + \rho_{K-E}  & \leq \alpha_{K} - \beta_{K-E} \\
0 & \leq \alpha_{k} - \beta_{k+1}, \ \forall k \in [K-E - 1]
\end{aligned} 
\end{equation}
which after the $(K-2)$-th elimination, boils down to
\begin{equation}
\label{eq:FM_elimination_2_1}
\begin{aligned}
\rho_{1} & \leq  \beta_{2} \\
\rho_{2} & \leq \alpha_{2} - \beta_{2} \\
\rho_{3}  + \rho_{2} & \leq \alpha_{3} - \beta_{2} \\
\vdots \\
\rho_{K} + \rho_{K-1} + \cdots + \rho_{2}  & \leq \alpha_{K} - \beta_{2} \\
0 & \leq \alpha_{1} - \beta_{2}.
\end{aligned} 
\end{equation}
Finally, we eliminate $\beta_{2}$ in \eqref{eq:FM_elimination_2_1}, from which we obtain
\begin{equation}
\label{eq:FM_elimination_final}
\begin{aligned}
\rho_{1} & \leq  \alpha_{1} \\
\rho_{2} + \rho_{1}  & \leq \alpha_{2}  \\
\vdots \\
\rho_{K} + \rho_{K-1} + \cdots + \rho_{2} +  \rho_{1}   & \leq \alpha_{K}.
\end{aligned} 
\end{equation}
It is evident that the set of inequalities in  \eqref{eq:FM_elimination_final} is identical to the set of inequalities that describe
$ \mathcal{D}^{\mathrm{PHY}}(\sigma,\bm{\alpha})$ in  \eqref{eq:GDoF_region_phy}, which completes the proof of Lemma \ref{lemma:FM_elimination}.
\subsection{Constant gap}
\label{appendix:subsec_constant_gap_PHY}
Here we show that the GDoF region characterization in Lemma \ref{lemma:FM_elimination}
translates to an approximate characterization of the capacity region.
The tools used to establish this result are reused further on in Appendix \ref{appendix:constant_gap}
to establish a similar result for the original cache-aided channel. 
\begin{corollary}
\label{corollary:constant_gap_PHY}
The capacity region $\mathcal{C}^{\mathrm{PHY}}(\sigma,\bm{\alpha},P)$ includes all non-negative rate 
tuples that satisfy
\begin{equation}
\label{eq:capacity_inner_phy}
\begin{aligned}
\sum_{i \in [k]} R_{i} + \sum_{\mathcal{S} \in \cup_{i \in [k]} \Sigma_{i} }  R_{\mathcal{S}}   & \leq \big( {\alpha_{k}} \log(P) - k \big)^{+},
\ \forall k \in [K - \sigma + 1] \\
\sum_{i \in [k]}  R_{i} + \sum_{\mathcal{S} \in \Sigma}  R_{\mathcal{S}}  & \leq \big( {\alpha_{k}} \log(P) - k \big)^{+}, \ \forall k \in [K - \sigma + 2 : K].
\end{aligned}
\end{equation}
Moreover, this achievable region is within $2$ bits (per dimension) from the entire capacity region for all system parameters.
That is, for any tuple $(R_{k}: k\in [K], \ R_{\mathcal{S}} : \mathcal{S} \in \Sigma)$  at the boundary of \eqref{eq:capacity_inner_phy}, the tuple $(R_{k} + 2: k\in [K], \ R_{\mathcal{S}} + 2 : \mathcal{S} \in \Sigma)$ is outside the capacity region $\mathcal{C}^{\mathrm{PHY}}(\sigma,\bm{\alpha},P)$.
\end{corollary}
\begin{proof} First, we observe from  \eqref{eq:GDoF_inner_beta} and Lemma \ref{lemma:FM_elimination}
that the achievable rate region $\mathcal{C}^{\mathrm{PHY}}_{\mathrm{in}}(\sigma,\bm{\alpha},P)$, described in 
\eqref{eq:capacity_inner_phy_beta},  is equivalently  expressed by the set of all non-negative rate tuples that satisfy
\begin{equation}
\label{eq:capacity_inner_phy_GDoF}
\begin{aligned}
R_{k} + \sum_{\mathcal{S} \in \Sigma_{k} } R_{\mathcal{S}}  & \leq 
\Big(  r_{k}  \log \left(P  \right) + \sum_{\mathcal{S} \in \Sigma_{k} } r_{\mathcal{S}}    \log \left(P  \right)  - 1 \Big)^{+}, \ \forall k \in [K - \sigma + 1]\\
R_{k}  & \leq \big( r_{k} \log \left(P  \right)  - 1 \big)^{+}, \ \forall k \in [K - \sigma + 2 : K]
\end{aligned}
\end{equation}
for some $(r_{k}: k\in [K], \ r_{\mathcal{S}} : \mathcal{S} \in \Sigma)  \in \mathcal{D}^{\mathrm{PHY}}(\sigma,\bm{\alpha})$.
Next, we show that \eqref{eq:capacity_inner_phy_GDoF} includes the achievable rate region describe in 
\eqref{eq:capacity_inner_phy}.
Suppose that $(R_{k}: k\in [K], \ R_{\mathcal{S}} : \mathcal{S} \in \Sigma)$ 
satisfies \eqref{eq:capacity_inner_phy} in 
Corollary \ref{corollary:constant_gap_PHY}.
From Lemma \ref{lemma:FM_elimination}, there must exists $\rho_{1}, \ldots, \rho_{K}$, as defined in \eqref{eq:rho_def},
such that 
\begin{equation}
\sum_{i \in [k]} R_{i} + \sum_{\mathcal{S} \in \cup_{i \in [k]} \Sigma_{i} }  R_{\mathcal{S}}   = 
\Big( \sum_{i \in [k]} \rho_{i} \log(P) - k \Big)^{+} \leq \big( \alpha_{k}\log(P) - k \big)^{+}, \ \forall k \in [K]
\end{equation}
where in the above, we have used $\Sigma_{i} = \emptyset$ for all $i \in [K - \sigma + 2 : K]$.
This implies that
\begin{align}
\nonumber
 R_{k} + \sum_{\mathcal{S} \in \cup_{k} \Sigma_{k} }  R_{\mathcal{S}} & = 
\left[  \sum_{i \in [k]} R_{i} + \sum_{\mathcal{S} \in \cup_{i \in [k]} \Sigma_{i} }  R_{\mathcal{S}} \right] - 
\left[ \sum_{i \in [k-1]} R_{i} + \sum_{\mathcal{S} \in \cup_{i \in [k-1]} \Sigma_{i} }  R_{\mathcal{S}} \right] \\
\nonumber
& \leq \Big( \sum_{i \in [k]} \rho_{i} \log(P) - k \Big)^{+} - \Big(  \sum_{i \in [k-1]} \rho_{i} \log(P) - (k-1) \Big)\\
& \leq  \big( \rho_{k} \log(P) - 1 \big)^{+}.
\end{align}
Therefore, $(R_{k}: k\in [K], \ R_{\mathcal{S}} : \mathcal{S} \in \Sigma)$ satisfies \eqref{eq:capacity_inner_phy_GDoF}, and hence is 
in $\mathcal{C}^{\mathrm{PHY}}_{\mathrm{in}}(\sigma,\bm{\alpha},P)$.

Now let us  consider a rate tuple 
$(R_{k}': k \in [K], R_{\mathcal{S}}': \mathcal{S} \in \Sigma )$
at the boundary of the  rate region in \eqref{eq:capacity_inner_phy}.
It follows that there exists some $k'$ in  $[K]$ such that \eqref{eq:capacity_inner_phy} holds with equality, that is
\begin{equation}
\label{eq:rate_tuple_equality}
\sum_{i \in [k']} R_{i}' + \sum_{\mathcal{S} \in \cup_{i \in [k']} \Sigma_{i} }  R_{\mathcal{S}}'   = \big( {\alpha_{k'}} \log(P) - k' \big)^{+}.
\end{equation}
Now consider  a second rate tuple given by
\begin{equation}
\label{eq:rate_tuple_double_prime}
(R_{k}'' = R_{k}' + 2: k \in [K], R_{\mathcal{S}}'' = R_{\mathcal{S}}'+ 2: \mathcal{S} \in \Sigma ). 
\end{equation}
For  the index $k'$ in \eqref{eq:rate_tuple_equality}, we have 
\begin{align}
\nonumber
\sum_{i \in [k']} R_{i}'' + \sum_{\mathcal{S} \in \cup_{i \in [k']} \Sigma_{i} }  R_{\mathcal{S}}''
& =  \sum_{i \in [k']} R_{i}' + \sum_{\mathcal{S} \in \cup_{i \in [k']} \Sigma_{i} }  R_{\mathcal{S}}'
+ 2 \cdot \Big( k' +  \sum_{i \in [k']}  |  \Sigma_{i} |  \Big) \\
\nonumber
& \geq {\alpha_{k'}}\log(P) - k'  + 2 \cdot \Big( k' +  \sum_{i \in [k']}  |  \Sigma_{i} |  \Big)  \\
\label{eq:rate_tuple_double_prime_outside}
& \geq {\alpha_{k'}}\log(P) + 1. 
\end{align}
The inequality in \eqref{eq:rate_tuple_double_prime_outside} implies that the rate 
tuple defined in  \eqref{eq:rate_tuple_double_prime} is not included in the outer bound $\mathcal{C}^{\mathrm{PHY}}_{\mathrm{out}}(\sigma,\bm{\alpha},P)$. This holds since the inequalities in \eqref{eq:capacity_outer_phy}  imply:
\begin{equation}
\label{eq:outer_ound_phy_strict}
\sum_{i \in [k]} R_{i} + \sum_{\mathcal{S} \in \cup_{i \in [k]} \Sigma_{i} }  R_{\mathcal{S}}   < \log(P^{\alpha_{k}}) + 1, \ \forall k \in [K]
\end{equation}
where strictness in the above inequalities is due to $P > 1$ and $\alpha_{k} > 0$, for all $k \in [K]$.
Moreover, as alluded to in Remark \ref{remark_GDoF_model}, for the regime $P \leq 1$, the all zero rate tuple is within one bit (per dimension) from all rate tuples in 
$\mathcal{C}^{\mathrm{PHY}}_{\mathrm{out}}(\sigma,\bm{\alpha},P)$.
This concludes the proof of Corollary \ref{corollary:constant_gap_PHY}.
\end{proof}
\section{Non-Integer $K \mu$}
\label{appendix:non_integer_K_mu}
In this appendix, we prove that the GNDT in \eqref{eq:tau_ub} is achievable for all $\mu$ such that $K \mu$ is non-integer.
Therefore, we assume throughout this appendix that $K \mu$ takes a non-integer value  in $(0,K)$.
Moreover, we focus on worst-case demands as defined in Section \ref{subsec:worst_case_demands}.
\subsection{Physical channel}
We first look at the physical channel problem.
In particular, let us consider a degraded GBC with three message sets: a unicast set, a $\sigma$-multicast set and a $\gamma$-multicast, where $\sigma, \gamma \in [2: K]$ and  $\sigma < \gamma$.
Similar to the definitions of $\Sigma$ and $\Sigma_{i}$ for the $\sigma$-multicast groups in Section \ref{subsec:unicast_sigma_multicast_messages}, we denote the set of all $\gamma$-multicast groups by $\Gamma$, which is partitioned into $\{ \Gamma_{i} : i \in [K - \gamma + 1] \}$.
The GDoF region of this channel is hence given by all  GDoF tuples of the form 
\begin{equation}
\label{eq:GDoF_region_phy_2_multicast}
\nonumber
\big(r_{k}: k\in [K], \ r_{\mathcal{S}} : \mathcal{S} \in \Sigma, \ r_{\mathcal{G}} : \mathcal{G} \in \Gamma \big) \in \mathbb{R}_{+}^{K + \binom{K}{\sigma} + \binom{K}{\gamma}}
\end{equation} 
which satisfy the following set of inequalities:
\begin{equation}
\sum_{i \in [k]} r_{i} +
 \sum_{\mathcal{S} \in \cup_{i \in [k]} \Sigma_{i} }  r_{\mathcal{S}}  +
 \sum_{\mathcal{G} \in \cup_{j \in [k]} \Gamma_{j} }  r_{\mathcal{G}}  \leq \alpha_{k}, \ \forall k \in [K]
\end{equation}
where we assume that $\Sigma_{i} = \emptyset$ for all 
$i \in [K - \sigma + 2 : K]$, and $\Gamma_{j} = \emptyset$ for all 
$j \in [K - \gamma + 2 : K]$.
The proof of \eqref{eq:GDoF_region_phy_2_multicast} follows the same steps used to prove Theorem \ref{theorem:GDoF_phy} in Appendix \ref{appendix:GDoF_region}, and is  omitted to avoid repetition.
Similar to Corollary \ref{corollary:symm_multicast_region}, the GDoF region in \eqref{eq:GDoF_region_phy_2_multicast}
yields the lower dimensional projection characterized by 
all tuples $(r_{k}: k\in [K], r_{\mathrm{sym}}^{\sigma} , r_{\mathrm{sym}}^{\gamma} ) \in \mathbb{R}_{+}^{K + 2}$
that satisfy
\begin{equation}
\label{eq:GDoF_region_phy_2_multicast_sym}
\sum_{i \in [k]} r_{i}  +  \left[ \binom{K}{\sigma} \! - \! \binom{K - \min\{ k,s \} }{\sigma} \right] \cdot r_{\mathrm{sym}}^{\sigma} +  
\left[ \binom{K}{\gamma} \! - \! \binom{K - \min \{ k, s \} }{\gamma} \right] \cdot 
r_{\mathrm{sym}}^{\gamma}  \leq \alpha_{k},  \forall k \in [K]
\end{equation}
which captures scenarios with symmetric $\sigma$-multicast  GDoF and symmetric $\gamma$-multicast GDoF, where each multicast message is intended to at least one users in $[s]$, for some $s \in [K]$.
\subsection{Caching and Delivery}
We introduce some notation which is used in the following parts. 
Recalling that $K \mu$ is non-integer, we set the multicast group sizes as: 
$\sigma  =  \lfloor K \mu + 1 \rfloor$ and $\gamma = \lceil K \mu + 1 \rceil$.
Moreover, we define the following values:  $\lambda  \triangleq  \gamma - (K \mu + 1) $ and $\bar{\lambda}  \triangleq (K \mu + 1) - \sigma$,  where it is evident that $\bar{\lambda} = 1 - \lambda $.

Content placement, preparing the coded multicast messages and recovering files at the receivers is carried out in 
the YMA manner, using the principle of memory-sharing \cite{Yu2018,Yu2019,Maddah-Ali2014}.
The problem reduces to delivering a set of $\sigma$-multicast messages, $\gamma$-multicast messages
and unicast messages.
It is worthwhile noting that due to file splitting during the placement phase and memory-sharing,  each of the $\sigma$-multicast messages has a normalized file size of $\lambda / \binom{K}{\sigma - 1}$, while each $\gamma$-multicast message
has a normalized file size of  $\bar{\lambda} / \binom{K}{\gamma - 1}$.
From \eqref{eq:GDoF_region_phy_2_multicast_sym},
it follows that for any achievable GDoF tuple $(\mathbf{r},r_{\mathrm{sym}}^{\sigma} , r_{\mathrm{sym}}^{\gamma})$,
where $\mathbf{r}$ comprises the GDoF of non-content messages,
a GNDT given by
\begin{equation}
\label{eq:GNDT_two_multicast}
\tau = \max \left\{ \frac{\lambda}{r_{\mathrm{sym}}^{\sigma} \cdot \binom{K}{\sigma - 1}} , 
 \frac{\bar{\lambda}}{r_{\mathrm{sym}}^{\gamma} \cdot \binom{K}{\gamma - 1}} \right\}
\end{equation}
is achievable. The point-wise maximum in \eqref{eq:GNDT_two_multicast} is due to the fact that the GNDT is determined by the \emph{slowest} group of coded content messages.
We further optimize the symmetric $\sigma$-multicast GDoF  and the symmetric  $\gamma$-multicast GDoF such that they satisfy
\begin{equation}
 r_{\mathrm{sym}}^{\sigma} = 
 r_{\mathrm{sym}}^{\gamma} \cdot   \frac{ \lambda  \binom{K}{\gamma - 1}}{\bar{\lambda} \binom{K}{\sigma - 1}}.
\end{equation}
In this case, \eqref{eq:GNDT_two_multicast} boils down to
\begin{equation}
\tau = \frac{\lambda}{r_{\mathrm{sym}}^{\sigma} \cdot \binom{K}{\sigma - 1}} = 
 \frac{\bar{\lambda}}{r_{\mathrm{sym}}^{\gamma} \cdot \binom{K}{\gamma - 1}}
\end{equation}
and for any achievable $\tau$, we achieve the GDoF region given by all $\mathbf{r} \in \mathbb{R}_{+}^{K}$
that satisfy 
\begin{equation}
\label{eq:GDoF_region_2_multicast_sym}
\frac{1}{\tau} \cdot  \left( \lambda \cdot \frac{\binom{K}{\sigma} - \binom{K - \min\{ k, N \} }{\sigma}}{\binom{K}{\sigma-1}}  + 
\bar{\lambda} \cdot \frac{ \binom{K}{\gamma} - \binom{K - \min\{ k, N \} }{\gamma} }{\binom{K}{\gamma-1}}   \right) + \sum_{i \in [k]} r_{i}   \leq \alpha_{k},  \ \forall k \in [K].
\end{equation}
This translates to an achievable GNDT of 
\begin{equation}
\label{eq:tau_convex_comb}
\tau = \max_{k \in [K]} \left\{ \frac{1}{\big( \alpha_{k} - \sum_{i \in [k]} r_{i} \big) } \cdot
\left( \lambda \cdot \frac{\binom{K}{\sigma} - \binom{K - \min\{ k, N \} }{\sigma}}{\binom{K}{\sigma-1}}  + 
\bar{\lambda} \cdot \frac{ \binom{K}{\gamma} - \binom{K - \min\{ k, N \} }{\gamma} }{\binom{K}{\gamma-1}}   \right) \right\}.
\end{equation}
Now it remains to show that \eqref{eq:tau_convex_comb} and \eqref{eq:tau_ub} 
are equal for all non-integer values of $K \mu$.

To this end, let us recall from \cite[Appendix J]{Yu2019} that for any $k \in [K]$, the sequence defined as
\begin{equation}
\label{eq:sequence_cn}
c_{n} \triangleq \frac{ \binom{K}{n+1} - \binom{K - \min\{k,N\} }{n+1}  }{  \binom{K}{n}  } 
\end{equation}
is convex  in $n \in [0 : K]$. 
Therefore, the points defined by \eqref{eq:sequence_cn} are corner points on their lower convex envelope 
given by  $f(n) = \mathrm{conv}(c_{n})$, and cannot be expressed as convex combinations of other points on $f(n)$ 
(see also \cite[Remark 7]{Yu2018}).
Hence for any non-integer value of $n$ in $(0,K)$, we have $f(n) = (\lceil n \rceil - n) f(\lfloor n \rfloor) + 
 (n - \lfloor n \rfloor)  f(\lceil n \rceil) $.
This is precisely the expression appearing inside the $\max \{ \cdot \}$ operator in \eqref{eq:tau_convex_comb}, 
from which it directly follows that \eqref{eq:tau_convex_comb} and \eqref{eq:tau_ub} are equal.
\subsection{Proof of \eqref{eq:tau_tilde_tau_breve}}
In the final part of this appendix, we show that the inequality in  \eqref{eq:tau_tilde_tau_breve} holds.
First, note that due to the convexity of the sequence in \eqref{eq:sequence_cn},
the new sequence given by 
\begin{equation}
C_{n} \triangleq \max_{k \in [K]} \left\{ 
\frac{ 1}{ \big( \alpha_{k} - \sum_{i \in [k]} r_{i} \big)}  \cdot 
\frac{ \binom{K}{n+1} - \binom{K - \min\{k,N\} }{n+1}  }{  \binom{K}{n}  }  \right\}
\end{equation}
is also convex in  $n \in [0 : K]$. This holds since  the point-wise maximum of convex functions is a convex function.
Therefore, the conclusions related to the lower convex envelope $f(n) = \mathrm{conv} (c_{n})$ above also hold for the 
lower convex envelope $F(n) = \mathrm{conv} (C_{n})$.

From the above, it follows that $\tau_{\mathrm{ms}}^{\mathrm{ub}} (\mathbf{r} ; \mu,  \bm{\alpha})  = \tau^{\mathrm{ub}} (\mathbf{r} ; \mu,  \bm{\alpha}) $ for all $\mu$ such that $K \mu$ takes integer values.
Moreover, for non-integer values of $K \mu$, the function $\tau_{\mathrm{ms}}^{\mathrm{ub}} (\mathbf{r} ; \mu,  \bm{\alpha})$ can be expressed as  
\begin{multline}
\label{eq:GNDT_ub_conv}
\tau_{\mathrm{ms}}^{\mathrm{ub}} (\mathbf{r} ; \mu,  \bm{\alpha})  = 
\lambda  \cdot \max_{k \in [K]} \left\{ 
\frac{ 1}{ \big( \alpha_{k} - \sum_{i \in [k]} r_{i} \big)}  \cdot 
\frac{ \binom{K}{\sigma} - \binom{K - \min\{k,N\} }{\sigma}  }{  \binom{K}{\sigma - 1}  }  \right\} \\
+ \bar{\lambda} \cdot \max_{k \in [K]} \left\{ 
\frac{ 1}{ \big( \alpha_{k} - \sum_{i \in [k]} r_{i} \big)}  \cdot 
\frac{ \binom{K}{\gamma} - \binom{K - \min\{k,N\} }{\gamma}  }{  \binom{K}{\gamma - 1}  }  \right\} \\
\geq  \max_{k \in [K]} \left\{ 
\frac{ 1}{ \big( \alpha_{k} - \sum_{i \in [k]} r_{i} \big) }  \cdot 
\left( \lambda \cdot \frac{ \binom{K}{\sigma} - \binom{K - \min\{k,N\} }{\sigma}  }{  \binom{K}{\sigma - 1}  } 
+ \bar{\lambda} \cdot \frac{ \binom{K}{\gamma} - \binom{K - \min\{k,N\} }{\gamma}  }{  \binom{K}{\gamma - 1}  } 
\right)
 \right\}.
\end{multline}
The inequity in \eqref{eq:GNDT_ub_conv}, which is identical to the inequality in 
\eqref{eq:tau_tilde_tau_breve}, is implied by  Jensen's inequality, as the   point-wise maximum function $\max \{ \cdot \}$ is convex in its arguments.
\section{Non-Worst-Case Demands}
\label{appendix:other_demands}
Here we show that the trade-off in \eqref{eq:tau_achievability}, shown to be achievable in Section \ref{sec:achievability}
when the weakest users make distinct demands, is also achieved whenever the distinct demands are not made by the weakest users.
In particular, we consider the case where $N < K$ with $N$ distinct demands, 
yet these distinct demands are not necessarily made by the first $N$ users. 
Whenever $N \geq K - \sigma + 1$, we transmit the set of all coded multicast messages and achieve the delay in \eqref{eq:tau_achievability}.
Therefore, we focus on the case where $N \leq K - \sigma$, in which not all coded multicast messages are transmitted.
Recall that placement is independent of user demands, and hence remains as in Section \ref{sec:achievability}.
\subsection{Coded multicast messages}
After user demands are revealed, 
we select the set of leading users as $\mathcal{U} = \{u_{1}, \ldots, u_{N} \}$,
such that $u_{1} \leq u_{2} \leq \cdots \leq u_{N}$, and each $u_{i}$ is the  weakest user (i.e. smallest index) that requests file  $F_{d_{u_{i}}}$.
Note that we must have $u_{1} = 1$.
The set of non-leading users is given by $\bar{\mathcal{U}} = [K] \setminus \mathcal{U}$.
Generating coded multicast messages is carried out as in the previous part,  in accordance with the YMA scheme, 
where each generated message is useful to at least one leading user.
The generated set of coded multicast messages is given by
$\{W_{\mathcal{S}} : \mathcal{S} \in \Sigma, \mathcal{S} \cap \mathcal{U} \neq \emptyset \}$.

Let us, for now, assume that leading users successfully decode their intended coded multicast messages, and hence recover their requested files.
We show that in this case, non-leading users will also be able to compute their missing coded multicast messages, and recover their requested files. 
\begin{lemma}
\label{lemma:non_leading_users_decode}
Given that each transmitted coded multicast message $W_{\mathcal{S}}$ 
is successfully decoded by all intended users in $\mathcal{S}$, then each non-leading user $k \in \bar{\mathcal{U}}$
can compute  all required missing coded multicast messages, i.e. $\{ W_\mathcal{A} : \mathcal{A} \subseteq \bar{\mathcal{U}}, | \mathcal{A} | = \sigma,
k \in \mathcal{A} \}$.
\end{lemma}
\begin{proof}
Consider an arbitrary missing coded multicast message $W_{\mathcal{A}}$, for some group of non-leading users $\mathcal{A} = \{a_{1},\dots, a_{\sigma} \} \subseteq \bar{\mathcal{U}}$, which we wish to compute.
We assume without loss of generality that $a_{1} = \min \{ \mathcal{A}\}$, i.e. the weakest user in the group $\mathcal{A}$.
To show that users in $\mathcal{A}$ can compute $W_{\mathcal{A}}$, it is sufficient to show that $a_{1}$ can compute $W_{\mathcal{A}}$. 
Next, we show that each of the transmitted coded multicast messages required for computing $W_{\mathcal{A}}$ is either intended to 
leading users which are no stronger than user $a_{1}$ or intended to user $a_{1}$; and hence decodable by all users in $\mathcal{A}$. 

To this end, let $u_{j}$ be the leading user that satisfies  $u_{j} < a_1 < u_{j+1}$.
If $j = N$, then $a_1$ is stronger than all leading users and hence can recover all their intended messages.
Combining this with the fact that $u_{1} = 1$,  we may proceed while assuming that there exists a pair of users $u_{j}$ and $u_{j+1}$ in  $\mathcal{U}$ such that $u_{j} < a_1 < u_{j+1}$ holds.
The file demanded by user $a_{1}$, i.e. $F_{d_{a_{1}}}$, must also be demanded by some user $u' \in \{ u_{1},\ldots , u_{j} \}$, since otherwise $a_{1}$ must be a leading user.
In reconstructing $W_{\mathcal{A}}$ according to \eqref{eq:W_A_YMA},
we define $\mathcal{B} = \{u_{1}, \ldots , u_{N}, a_{1}, \ldots, a_{\sigma} \}$,
and $\Upsilon$ as the family of subsets of $\mathcal{B}$ that constitute potential sets of leaders, other than $\mathcal{U}$.

The problem reduces to showing that each $\mathcal{B} \setminus \mathcal{V}$, where $ \mathcal{V} \in \Upsilon$, contains at least one user from  $\{ u_{1}, \ldots, u_{j}, a_{1} \}$.
To show this, first consider the case where $a_1 \in \mathcal{V}$. Here we must have $u' \notin \mathcal{V}$ and therefore $u' \in \mathcal{B} \setminus \mathcal{V}$, which proves the statement in Lemma \ref{lemma:non_leading_users_decode}.
Now let us consider the second case where $a_1 \notin \mathcal{V}$.
If  $u' \notin \mathcal{V}$ also holds, i.e. there is a third user demanding $F_{d_{a_{1}}}$ and is in $\mathcal{V}$, then Lemma  \ref{lemma:non_leading_users_decode} holds, and therefore we focus on the remaining case where $u' \in \mathcal{V}$.
In this case, $a_{1}$ cannot be in $\mathcal{V}$, and hence we must have 
$a_{1} \in \mathcal{B} \setminus \mathcal{V}$, which completes the proof.
\end{proof}
It is worthwhile highlighting that an observation similar to Lemma \ref{lemma:non_leading_users_decode} was made in \cite[Rem. 2]{Amiri2018b}.
Following Lemma \ref{lemma:non_leading_users_decode}, we  now focus on the transmission of  the sets of messages 
$\{W_{k} : k \in  [K] \}$ and $\{W_{\mathcal{S}} : \mathcal{S} \in \Sigma, \mathcal{S} \cap \mathcal{U} \neq \emptyset \}$, and characterize the corresponding achievable performance.
\subsection{Transmission}
The sets of coded multicast messages and unicast messages are transmitted using
the physical-layer scheme in Section  \ref{sec:degraded_GBC_unicast_multicast}.
Next, we show that \eqref{eq:tau_achievability} is also achievable in this case by deriving an upper bound on the achievable GNDT, given any achievable GDoF tuple $\mathbf{r}$, which matches the one in 
\eqref{eq:tau_achievability}.
To this end, we define a subset of $\sigma$-multicast groups given by
\begin{equation}
\Sigma' = \big\{ \mathcal{S} \in \Sigma : \mathcal{S} \cap \mathcal{U} \neq \emptyset \big\}
\end{equation}
comprising all  multicast groups that include at least one leading user.
It follows that the set of coded multicast messages of interest (i.e. to be transmitted) is given by 
\begin{equation}
\big\{W_{\mathcal{S}} : \mathcal{S} \in \Sigma,  \mathcal{S} \cap \mathcal{U} \neq \emptyset  \big\} = 
\big\{ W_{\mathcal{S}} : \mathcal{S} \in \Sigma' \big\}.
\end{equation} 
Since  $\Sigma'$ is equal to $\Sigma \setminus \big\{ \mathcal{S} \in \Sigma : \mathcal{S} \cap \mathcal{U} = \emptyset \big\}$, the cardinality of $\Sigma' $ is given by
\begin{equation}
\big| \Sigma'  \big| = \binom{K}{\sigma} - \binom{K - N}{\sigma}.
\end{equation}
Moreover, $\Sigma'$ can be paritioned into the family  $\big\{ \{\Sigma_i \cap \Sigma' \} : i \in [ K ]  \big\}$, where $\big\{ \Sigma_i : i \in [ K ]  \big\}$ is the partition of $\Sigma$ defined in Section \ref{subsec:unicast_sigma_multicast_messages}. Recall that  $\Sigma_i = \emptyset$ for all $i > K -\sigma +1$ (see Remark \ref{remark:Sigma_augment}).

We now focus on transmission over the physical channel based on the scheme in Section \ref{sec:degraded_GBC_unicast_multicast}. 
Recall that since we have assumed (without loss of generality) that leading users are ordered as 
$1 = u_{1} \leq u_{2} \leq \cdots \leq u_{N}$, then  $i \leq u_i$ must hold for all $i \in [N]$.
By setting the achievable GDoF for all missing $\sigma$-multicast messages to zero, and restricting to a symmetric GDoF across remaining  messages, the GDoF region in Theorem \ref{theorem:GDoF_phy} becomes
\begin{equation}
\label{eq:GDoF_region_PHY_missing}
\begin{aligned}
\sum_{i \in [k]} r_{i}  + 
\bigg| \bigcup_{i \in [k]} 
\big\{ \Sigma_{i} \cap \Sigma' \big\}  \bigg| \cdot  r_{\mathrm{sym}}   &  \leq \alpha_{k}, \ \forall k \in [u_{N} ] \\
\sum_{i \in [k]} r_{i} + \big|  \Sigma'  \big| \cdot  r_{\mathrm{sym}}    & \leq \alpha_{k}, \ \forall k \in [u_{N} + 1: K].
\end{aligned}
\end{equation}
Next, we wish to obtain a more tractable inner bound for the region in \eqref{eq:GDoF_region_PHY_missing}. 
We observe that  for all
$k \in [N]$, we have  $\big| \cup_{i \in [k]} 
\{ \Sigma_{i} \cap \Sigma' \}  \big| \leq \big| \cup_{i \in [k]}  \Sigma_{i} \big| = \binom{K}{\sigma} - \binom{K - k}{\sigma}$.
On the other hand, for all $k \in [N+1:u_N]$,  we can write $\big| \cup_{i \in [k]}  \{ \Sigma_{i} \cap \Sigma' \}  \big| \leq \big|  \Sigma'   \big| = \binom{K}{\sigma} - \binom{K - N}{\sigma}$.
It follows that the symmetric $\sigma$-multicast GDoF region in \eqref{eq:GDoF_region_PHY_missing} includes the achievable region given by
\begin{equation}
\label{eq:GDoF_region_PHY_missing_2}
\begin{aligned}
\sum_{i \in [k]} r_{i} + \bigg[  \binom{K}{\sigma} - \binom{K - k}{\sigma} \bigg] \cdot  r_{\mathrm{sym}}   & \leq \alpha_{k}, \ \forall k \in [N]\\
\sum_{i \in [k]}  r_{i} + \bigg[  \binom{K}{\sigma} - \binom{K - N}{\sigma} \bigg]\cdot  r_{\mathrm{sym}}   & \leq \alpha_{k}, \ \forall k \in [N+1: K].
\end{aligned}
\end{equation}
The above holds as for each $k \in [K]$, the corresponding inequality in \eqref{eq:GDoF_region_PHY_missing_2} implies its counterpart inequality in 
\eqref{eq:GDoF_region_PHY_missing}.
Therefore, it follows from \eqref{eq:GDoF_region_PHY_missing_2} that for any feasible unicast GDoF tuple $\mathbf{r} = (r_{k} : k\in [K])$, 
an achievable symmetric multicast GDoF is give by
\begin{equation}
\label{eq:max_r_sym_2}
r_{\mathrm{sym}} \leq \min_{k \in [K]} \left\{ \frac{ \big( \alpha_{k} - \sum_{i \in [k]} r_{i} \big) }{ \binom{K}{\sigma} - \binom{K - \min\{k,N\}}{\sigma}  }  \right\}
\end{equation}
from which we conclude that the GNDT-GDoF trade-off in \eqref{eq:tau_achievability}  is achievable in this case as well.
\section{Approximate Delay-Rate Characterization}
\label{appendix:constant_gap}
In this appendix, we show that the optimal GNDT-GDoF characterization in Corollary \ref{corollary:GDoF_region} (and Theorem \ref{theorem:optimal_trade_off}) leads to an approximate optimal delay-rate  characterization, as stated in Remark \ref{remark:delay_rate_char}.
\begin{corollary}
The capacity region $\mathcal{C}(\mathcal{T} ; \mu , \bm{\alpha}, P)$ includes all non-negative rate 
tuples  that satisfy
\begin{equation}
\label{eq:capacity_inner_caching}
\sum_{i \in [k]} R_{i} + \frac{1}{\mathcal{T}} \cdot \mathrm{conv} \left( \frac{ \binom{K}{K \mu +1} - \binom{K - \min\{k,N\}  }{K \mu + 1}  }{ \binom{K}{K \mu} }  \right)   \leq \big( {\alpha_{k}} \log(P) - k \big)^{+},  \forall k \in [K].
\end{equation}
Moreover, for any delay-rate tuple $(\mathcal{T}, \mathbf{R}' ; \mu)$ such that $\mathbf{R}' = (R_{1}',\ldots , R_{K}')$ is at the boundary of 
the achievable region described in \eqref{eq:capacity_inner_caching}, 
the best any scheme can do is to increase each rate $R_{k}'$ by less than $2$ bits per channel use, and reduce  $\mathcal{T}'$ by at most a multiplicative factor of $2.01$.
\end{corollary}
\begin{proof}
Following the same steps in Appendix \ref{appendix:subsec_constant_gap_PHY}, and combining with the achievability arguments in Section \ref{sec:achievability} and Appendix \ref{appendix:non_integer_K_mu}, it can be verified that $\mathcal{C}(\mathcal{T} ; \mu, \bm{\alpha}, P)$ includes the achievable rate region described above in \eqref{eq:capacity_inner_caching}.
On the other hand, the same argument used to show  \eqref{eq:outer_ound_phy_strict} in Appendix \ref{appendix:subsec_constant_gap_PHY}
can be employed to show that the outer bound derived in Section \ref{section:converse} implies that 
any rate tuple in $\mathcal{C}(\mathcal{T} ; \mu ,  \bm{\alpha},  P)$ must satisfy
\begin{equation}
\label{eq:outer_bound_strict_caching}
 \sum_{i \in [k]} R_{i} + \frac{1}{2.01 \cdot\mathcal{T} } \cdot \mathrm{conv} \left(  \frac{\binom{K}{K \mu +1} - \binom{K-\min\{k,N\} }{K \mu +1}}{\binom{K}{K \mu}} \right)  < \alpha_{k}  \log(P) + 1,  \forall k \in [K].
\end{equation}
Now let us introduce the rate tuple  $\mathbf{R}'' = (R_{1}'', \ldots , R_{K}'') = (R_{1}' + 2, \ldots , R_{K}' + 2)$ and 
the delay $\mathcal{T}'' = \mathcal{T} / 2.01$.
Since $\mathbf{R}'$ is at the boundary of the region in \eqref{eq:capacity_inner_caching}, we must have an index $k' \in [K]$
such that at least one of the inequalities in \eqref{eq:capacity_inner_caching} holds with equality. 
This implies that
\begin{align}
\nonumber
 \sum_{i \in [k']} R_{i}'' + \frac{1}{2.01 \cdot \mathcal{T}''} \cdot \mathrm{conv} &  \left(  \frac{\binom{K}{K \mu +1} - \binom{K- \min\{k',N\}  }{K \mu +1}}{\binom{K}{K \mu}} \right)  \\
\nonumber
&  =  \sum_{i \in [k']} R_{i}' + \frac{1}{\mathcal{T}'} \cdot \mathrm{conv} \left(  \frac{\binom{K}{K \mu +1} - \binom{K-\min\{k',N\} }{K \mu +1}}{\binom{K}{K \mu}} \right)   + 2k' \\
& = \big( {\alpha_{k'}} \log(P) - k' \big)^{+} +  2k'  \\
\label{eq:outer_bound_strict_caching_2}
& \geq {\alpha_{k'}} \log(P)  +  1.
\end{align}
It follows from \eqref{eq:outer_bound_strict_caching} and \eqref{eq:outer_bound_strict_caching_2} that 
$\mathbf{R}''$ is strictly outside the capacity region $\mathcal{C}(\mathcal{T}'' ; \mu, \bm{\alpha},  P)$,
and therefore the delay-rate tuple $(\mathcal{T}'', \mathbf{R}'' ; \mu)$
is in fact not achievable.
\end{proof}
\section*{Acknowledgements}
The authors would like to thank the anonymous reviewers for their valuable comments, which helped improve the quality of this paper.
\bibliographystyle{IEEEtran}
\bibliography{References}
\end{document}